\documentclass[11pt,A4,showpacs,twocolumn]{revtex4-1}
\usepackage{amsmath,amsfonts,amssymb}
\usepackage{amsthm}
\usepackage{graphicx}
\usepackage{hyperref}
\usepackage{pstricks}

\usepackage{amsmath}
\DeclareMathOperator\atanh{arctanh}

\def\be{\begin{equation}}
\def\ee{\end{equation}}

\theoremstyle{definition} \newtheorem{defi}{Definition}
\theoremstyle{plain} \newtheorem{theo}{Theorem}
\theoremstyle{plain} \newtheorem{lem}{Lemma}
\theoremstyle{plain} \newtheorem{prop}{Proposition}
\theoremstyle{plain} \newtheorem{cor}{Corollary}
\theoremstyle{remark} 

\newcommand{\arctanh}[1]{\mathrm{arctanh}#1}

\def\<{\langle}
\def\>{\rangle}

\newcommand{\Ham}{\mathcal{H}}
\newcommand{\B}{\beta}
\def\xr{{\underline{x}_r}}
\def\sr{{\underline{s}_r}}
\def\xa{{\underline{x}_a}}
\def\xz{{\underline{x}_z}}
\def\xg{{\underline{x}_\gamma}}
\def\allx{{\underline{x}}}
\def\alls{{\underline{s}}}
\def\sfour{{\underline{s}_P}}
\def\rinregs{{r \in R}}
\def\Ac{A^o}
\def\Acrc{ A^o_{r_0}}
\def\Arc{ A_{r_0}}
\def\Bra{B_{r_0,\alpha}}

\def\Braco{\bar{B}_{r_0,\alpha}^o}
\def\ptc{P-t-C}

\def\mra{m_{r_0\to\alpha}}
\def\lub{\mbox{lub}}
\def\hatu{\hat u}
\def\hatup{\hat u_P}
\def\hatU{\hat U}
\def\matstab{\mathbb{ K}}
\def\fcvm{F_{\text{\tiny{CVM}}}}

\begin{document}
\title{Gauge-free cluster variational method by maximal messages and moment matching}

\author{Eduardo Dom\'inguez}
\affiliation{Department of Theoretical Physics, Physics Faculty,
  University of Havana, La Habana, CP 10400, Cuba. }
\affiliation{``Henri Poincar\'e'' group of Complex Systems, University of Havana, Cuba. }
  
\author{Alejandro Lage-Castellanos}
\affiliation{Department of Theoretical Physics, Physics Faculty,
  University of Havana, La Habana, CP 10400, Cuba. }
\affiliation{``Henri Poincar\'e'' group of Complex Systems, University of Havana, Cuba. }
\affiliation{CNRS, Laboratoire de Physique Statistique, {\'E}cole Normale Supérieure, 75005, Paris.}

  \author{ Roberto Mulet}
\affiliation{Department of Theoretical Physics, Physics Faculty,
  University of Havana, La Habana, CP 10400, Cuba. }
\affiliation{``Henri Poincar\'e'' group of Complex Systems, University of Havana, Cuba. }

\author{Federico Ricci-Tersenghi} \affiliation{Dipartimento di Fisica,
  INFN -- Sezione di Roma 1 and CNR -- Nanotec, unit\`a di Roma,\\
  Universit\`{a} La Sapienza, P.le A. Moro 5, 00185 Roma, Italy}

\date{\today}

\begin{abstract}

We present a new implementation of the Cluster Variational Method (CVM) as a message passing algorithm.
The kind of message passing algorithms used for CVM, usually named Generalized Belief Propagation, are a generalization of the Belief Propagation algorithm in the same way that CVM is a generalization of the Bethe approximation for estimating the partition function.
However, the connection between fixed points of GBP and the extremal points of the CVM free-energy is usually not a one-to-one correspondence, because of the existence of a gauge transformation involving the GBP messages.

Our contribution is twofold. Firstly we propose a new way of defining messages (fields) in a generic CVM approximation, such that messages arrive on a given region from all its ancestors, and not only from its direct parents, as in the standard Parent-to-Child GBP. We call this approach {\it maximal messages}.
Secondly we focus on the case of binary variables, re-interpreting the messages as fields enforcing the consistency between the moments of the local (marginal) probability distributions.
We provide a precise rule to enforce all consistencies, avoiding any redundancy, that would otherwise lead to a gauge transformation on the messages.
This {\it moment matching} method is gauge free, {\it i.e.} it guarantees that the resulting GBP is not gauge invariant.

We apply our maximal messages and moment matching GBP to obtain an analytical expression for the critical temperature of the Ising model in general dimensions at the level of plaquette-CVM. The values obtained outperform Bethe estimates, and are comparable with loop corrected Belief Propagation equations. 
The method allows for a straightforward generalization to disordered systems.

\end{abstract}

\maketitle

\section{Introduction}

The Ising ferromagnet is one of the most studied and celebrated models in Statistical Physics. Although it lacks a proper analytical solution in three dimensions, it is globally well understood \cite{Huang}. However, the addition of disorder to this model generates a more complex scenario. Roughly speaking, the low temperature phase of the disordered model is not composed any more by two equivalent ordered phases as in the pure ferromagnetic model, but by many disordered phases with a complex structure. 
Techniques like the replica trick \cite{MPV} and the cavity method \cite{MP1,MP2} opened the door to the analytical treatment of the disordered variants of this and similar models in fully connected or in locally tree-like random graphs \cite{MezMonbook}.

However, finite dimensional systems remain a challenging problem regarding the analytical solutions. Only recently \cite{yedidia,tommaso_CVM,average_mulet,zhou_11,zhou_12,eduardoRFIM} has been realized that a proper generalization of the Bethe approximation, known with the name of Cluster Variational Method (CVM), could be a good starting point for a systematic treatment of these kind of disordered problems. The main task is to translate the (approximate) free energy saddle point conditions in a set of message passing equations, that can be solved efficiently even on large systems.

The interest in this kind of approximations is not only theoretical, but it comes also from many applications. For example, in image processing \cite{tanakaCVM1995,tanakaCVM,tanaka_replica_cvm,gouillart_discrete_tom}, 
it is important to improve the quality of the reconstruction algorithms, and message passing derived from CVM approximations has proved to be a good candidate in this direction \cite{tanaka_morita_CVM_image95}. Error correction and LDPC codes is another example of applications where GBP has been studied \cite{sibel12}, including the idea of fixing the gauge \cite{Pakzad2005Kikuchi}.





Most of previous works on cluster variational method and replica method, relied on the so called Parent-to-Child message passing \cite{yedidia},
which consists of an extension of the belief propagation for the Bethe approximation to more involved region graph approximations of the free 
energy. It has been shown that the Parent-to-Child message passing is redundant \cite{Pakzad2005Kikuchi,zhou_redundant,GBPGF}, since it introduces more 
``cavity'' fields (messages) than actually needed, producing a sort of gauge invariance in the solution. According to our experience, 
this invariance is not a big problem in the implementation of message  passing algorithms on a given finite dimensional instance, but
it certainly is a waist of computational resources, since more parameters need to be implemented.
In \cite{Aurelien16}, however, authors reports the gauge invariance as causing convergence problems. In any case this gauge invariance may obscure 
the connection between the average case prediction of the CVM equations derived for a disordered model and the solution of message passing 
equations in single instances of the model \cite{average_mulet}. To alleviate this problem we propose a general procedure to generate 
Gauge Free Generalized Belief Propagation (GFGBP) algorithm starting from a Cluster Variational Method approximation.

The procedure developed consists of the following steps:
\begin{enumerate}
\item definition of {\bf maximal messages}, in place of parent-to-child  messages 
\item definition of {\bf moment matching fields} in place of Lagrange multipliers to ensure beliefs consistency.
\end{enumerate}
The first is just another possible choice of messages that is quite general. 
The latter is a change of perspective in the interpretation of messages as Lagrange multipliers forcing marginalizations, to fields forcing consistency 
of moments in the beliefs distributions. This allows a systematic construction of gauge free message passing for any model with binary variables.

In \cite{zhou_redundant, Pakzad2005Kikuchi} authors developed a way to remove the redundancy in the GBP equations by removing the redundant messages. 
Our approach differs from theirs in that they keep with the Parent-to-Child approach of \cite{yedidia} and propose to fix the gauge by removing 
some messages completely from the belief expression of given regions in order to avoid loops in the region graph representation. We, instead, propose a larger
set of messages, but with properly reduced degrees of freedom.


We will apply the gauge free approach to the computation of critical temperatures in the plaquette-CVM approximation in Ising model in 
general dimensions, obtaining analytical expressions that improve over Bethe. The high dimension expansion of the critical temperature 
is correct until the third order term, as is the loop calculus of Ref.~\cite{loop_corrected_tom}. We also test the procedure in single 
instance implementation of message passing in Ising model. The more complicated (and interesting) disordered models, are left for future work.


The paper is organized as follows. Sec.~\ref{sec:CVM} introduces CVM and message passing algorithm in general terms, while Sec.~\ref{sec:maximalmessages} explains the maximal messages (MM) and the moment matching (MM) approaches; finally in Sec.~\ref{sec:isingd} we apply the MM-MM CVM (or 4M-CVM in short) algorithm to the calculation of the critical temperature in Ising models of general dimensions at the plaquette level.
For the sake of readability, we defer to the appendices the technical proofs.

\section{Cluster Variational Method} \label{sec:CVM}

The kind of problems we are dealing with are those statistical mechanics problems that require the computation of the properties of a large set of binary variables $x_i \in\{1,-1\}$, whose joint probability distribution
\begin{equation}
P(\allx) = \frac 1 Z \exp(-\beta \Ham(\allx)) \label{eq:Px_boltzmann}
\end{equation}
depends on a Hamiltonian $\Ham(\allx)$ that can be written as the sum of local terms
\[
\Ham(\allx) = \sum_a E_a(\xa)\,,
\]
where every interaction ``$a$'' with energy $E_a(\xa)$ involves a small subset of variables $\allx_a$.
This also includes the case of Bayesian networks, and therefore of many interesting inference problems. 

Computations of the statistical properties of 
each variable $x_i$ or groups of them, face the numerical difficulty of tracing over an exponential number of configurations when marginalizing over the 
remaining variables, and in general approximations are required. In the case of mean field, Bethe, and region graph approximations (see  \cite{idiosync_yedidia}), 
the underlying idea is to factorize the full probability distribution $P(\allx)$ into many smaller distributions containing a non 
extensive number of variables that we will refer to as regions.

\begin{figure*}[!htb]
  \includegraphics[width=\textwidth]{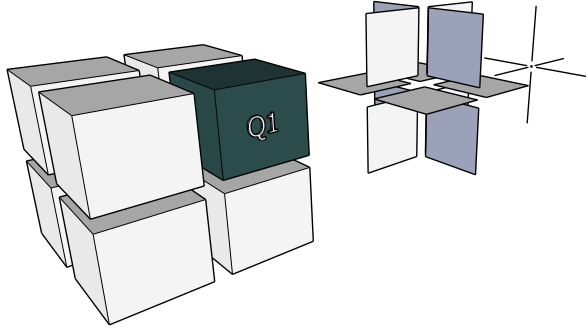}
\caption{ Example of regions surrounding the central spin $s_1$ in the cube approximation for the 3D square lattice model. Left: the 8 cubic 
regions $Q_1,\ldots,Q_8$. Center: the 12 faces (plaquettes) $P_1,\ldots,P_{12}$ shared by the maximal cubic regions. Right: the 6 vertex shared by 
the plaquettes and the central spin. The cube $Q_1$ is highlighted for later use. \label{fig:cube3D}}
\end{figure*}

The CVM \cite{kikuchi,yedidia} starts from a set of maximal regions $R_0$ (basic clusters), where no region is subset of another, and constructs a 
hierarchy of regions over which the approximation is defined. We will require that each degree of freedom $x_i$ and also all interactions $E_a(\xa)$ 
are present in at least one of these regions. Then we extend $R_0$ with the closure under the intersection operation as explained next.

From $R_0$, we define recursively the set of intersections $R_k$ as
\[ R_k = \{ r = r_{k-1}  \cap r_{k-1}' | r_{k-1},r_{k-1}' \in R_{k-1}\}\]
The whole group of regions is $R = R_0\cup R_1\cup R_2 \ldots$. Actually, in the CVM construction, the same regions might appear
more than once, and in different levels of intersections. Regardless this degeneracy, the relevant set is  $R$, the collection of
all regions obtained. Of utmost importance for later proves is that $R$ is a 
closed set under intersections and a partially ordered set, in which the subset relation defines the partial order, and 
$R_0$ is the set of maximal regions.

Since the system Hamiltonian is given by sums of local interactions between subsets of variables, 
we will consider that every time that the set of variables $\xa$ are part of a given region $\xa \subset r$, then the interaction ``$a$'' itself is 
part of it, allowing us to define the energy of the region as:
\[ E_r(\xr)  = \sum_{a\in r} E_a(\xa)\]
Since all interactions are at least part of one maximal region $r_0\in R_0$, we can write the Hamiltonian of the system as a sum over regions:
\begin{equation}
\Ham = \sum_{\rinregs} c_r E_r(\xr) \label{eq:E_r} 
\end{equation}
where the counting numbers $c_r$ guarantee that every interaction is counted exactly once \cite{yedidia}:
\begin{equation}
c_\alpha = 1-\sum_{r\in A_\alpha} c_r. \label{eq:c_r}
\end{equation}
The set $A_\alpha$ stand for the set of all ancestors of region $\alpha$, this is all super-regions of region  $\alpha$
\[A_\alpha = \{r\in R | \alpha \subset r \}.\]

Before going further in detail, let us visualize an example of the regions generated by the CVM construction. 
Consider a 3-dimensional spin model, with spins living in the nodes of a 3D-square lattice. In the cubic approximation, 
maximal regions are taken as the basic cubic cell of the lattice, with all its eight degrees of freedom at the cube's vertex. 
As a representative part of the full system, diagrams in figure \ref{fig:cube3D}
show all the regions containing the central spin $s_1$ (depicted as the central point in the rightmost diagram). Notice that the intersections of
the cubic regions in the leftmost diagram produce the square plaquette regions in the center diagram, with a spin at every angle of the squares. And the intersection of those plaquettes 
result in the rod (edges with two spins) regions in the rightmost diagram, which intersect only in the central spin.

\subsection{Variational approach and message passing}
Next we reproduce the approach by Zhou {\it et al}  \cite{zhou_redundant} on the derivation of message passing equations, instead of that of 
Yedidia  \cite{yedidia}. We prefer the former because it is
somehow more direct in the choice of the belief equations, saving the time of passing through Lagrange multipliers.

It starts by noting that, in accordance with eq. (\ref{eq:E_r}), the exact partition function of a system 
can be written as:
\begin{eqnarray*}
 Z(\B) &\equiv& \sum_\allx \exp(-\B \Ham(\allx)) \\
&=&  \sum_\allx \prod_{\rinregs }  \left[\exp(-\B  E_r(\xr)) \right]^{c_r} 
\end{eqnarray*}
 
A set of non zero test functions $\{m_{z}(\xz)\}$ can be multiplied and divided in the right hand side, 
such that they cancel out. We will call these test functions, messages.
Let us define $\partial z \subset R$ the set of regions in which the message $m_{z}(\xz)$ appears, 
and let $D_r $ be the set of messages entering region $r$, then 
\begin{eqnarray*} 
 Z(\B) &=& \sum_\allx \prod_{\rinregs }  \left[\exp(-\B  E_r(\xr))  \prod_{z \in D_r}  m_z(\xz) \right]^{c_r} 
\end{eqnarray*}
will still be the same partition function (independently of the values of the messages) if:
\begin{equation} 
\forall z\; \quad \sum_{r \in \partial z} c_r =0 \:\: .
\label{eq:mc_r}
\end{equation}

We can write an approximation to the free energy of the model in terms of the local beliefs
\begin{equation}
 b(\xr) = \frac 1 {z_r} \exp(-\B  E_r(\xr))  \prod_{z \in D_r}  m_z(\xz) \label{eq:br}
\end{equation}
with local partition functions
\[ z_r = \sum_{\xr} \exp(-\B  E_r(\xr))  \prod_{z \in D_r}  m_z(\xz).
\]
The free energy of the model $F = -kT \log Z(\B)$ can be rewritten as:
\begin{eqnarray*}
 F &=& -k T \sum_{\rinregs} c_r \log z_r -k T \log \left[ \sum_\allx \prod_{\rinregs } b_r(\xr)^{c_r}  \right]  \\
 &=& F_{R} + \Delta F
\end{eqnarray*}
This expression is still exact (independently of the value of the message functions). We will regard the first term 
$F_R[\{m\}]$ as a variational approximate to the real free energy. The rationale for this goes as follows. It can be shown (and will be)
that the minimization of the first term is equivalent to imposing local consistency between marginals of the beliefs functions (those appearing
in the second term). Once the beliefs are locally consistent it can be shown that the 
correct joint probability distribution of the model  $P(\underline x)$ can be written in a factorized form
as $\prod_{\rinregs } b_r(\xr)^{c_r}$ as far as the underlaying graph is a tree, therefore $\Delta F = - K T \log 1 = 0$.
This proves, {\it en passant}, that the approximation is exact for the case of tree topologies.
A rigorous justification of the approximation is absent, but in \cite{zhou_12} authors relate $\Delta F$
to the sum of correction contributions in the loop expansion of the free energy.
In the general case, i.e. loopy graphs, working with locally consistent beliefs that follow from the extremization of the first term, does not guarantees 
that the factorized measure $\prod_{\rinregs } b_r(\xr)^{c_r}$ is properly normalized, therefore at the fixed point $\Delta F \neq 0$ generally. 
Nevertheless, in all the situations where is meaningful to use message passing algorithms, we expect the corrections due to loops to be small, and
for this very reason, also $\Delta F \ll F_R$.


As a consequence of the variational treatment, we now need to solve the set of equations 
\begin{equation} \label{eq:dF_dm} 
\frac{\partial F_R[\{m\}]}{\partial m_{z}} = 0 \quad \forall z \;.
\end{equation}
The precise form of the resulting equations, and what exactly are they enforcing depends on the 
choice made for the messages and how do they appear in the belief equations. Next we explain one possible choice that we retain as the 
natural one and  we will call maximal message passing.

\section{Gauge-free 4M-CVM: maximal messages and moment matching}
\label{sec:maximalmessages}

Previously \cite{yedidia,GBPGF,zhou_redundant,Pakzad2005Kikuchi,Aurelien16} the set of messages have been defined using the so called Parent-to-Child (\ptc) approach.
This means that messages $m_{\alpha \to \gamma}(\xg) $ are indexed by two regions labels, the father one $\alpha$ and the child one $\gamma$. 
We will say that a region  $\alpha$ is father of $\gamma$, if $\alpha \supset \gamma$ and no region in $R$ is a subset of $\alpha$ and a 
superset of $\gamma$. In \ptc\   no messages are considered from grandparents or higher ancestors.

While this approach is very systematic, it has the problem of introducing too many degrees of freedom in the test functions. 
As already mentioned this may not have major consequences (besides efficiency) in physical observables measured on a given instance, 
but introduces a gauge invariance that might be problematic in the comparison with the typical behavior of message passing equations in the average case scenario  \cite{GBPGF,average_mulet}.
The reason is that in population dynamics one assumes messages arriving on a given region from different ancestors to be mostly uncorrelated: there are situations where this approximation is physically valid (e.g. when correlations are not too strong and regions are large enough), however the gauge invariance implies messages can freely change under the gauge transformation and this introduces undesirable correlations among the messages. For this reason a scheme free from the gauge invariance is very welcome.


We propose a top-down approach, that we call \emph{maximal message} passing, in which messages to region $r$ flow from all its ancestors 
$p\supset r$. 
We prefer the maximal messages, among other possibilities because it will allow us later to construct a gauge free system of message passing equations.

\begin{defi}{Maximal messages}
are defined by the set of message functions used, and by how these functions participate in each regions belief, as follows:
\begin{itemize}
\item every region $r\in R$  receives messages from all its ancestors $p \in A_r$. 
Messages are functions of the degrees of freedom in $r$: $m_{p \to r} (\xr) $,
 \item message $m_{p \to \gamma} (\underline{x}_\gamma) $ will appear in the region partition function $z_{r}$ 
 (i.e. equation  (\ref{eq:br})) of region $r$, iff  $ r \cap p = \gamma $. 
 \end{itemize}
\end{defi}
The first point asserts that there are as many message functions as pairs of comparable regions in the CVM construction, and therefore the 
messages are labeled by these two regions as $m_{p\to r}(\xr)$ where $r\subset p$.
The term {\it maximal messages} comes since there are messages to every region from all its ancestors, not only the frist direct parents as 
in parent to child. The second point defines $D_r$, the set of messages  $m_{p \to \gamma} (\underline{x}_\gamma)$ that appear 
multiplicatively in  the belief expression (or the partition function) of a certain region $r$, as
\[D_r = \{ p,\gamma \in R | p \nsubseteq r,  p \cap r =\gamma \neq \varnothing \}. 
\]
Put together, we have an expression for the beliefs at any region given by:
\begin{equation}
 b(\xr) = \frac 1 {z_r} e^{-\B  E_r(\xr)}  \prod_{p,\gamma \in D_ r}  m_{p \to \gamma}(\xg) \label{eq:brgl}
\end{equation}
\begin{figure*}[!htb]
  \includegraphics[width=\textwidth]{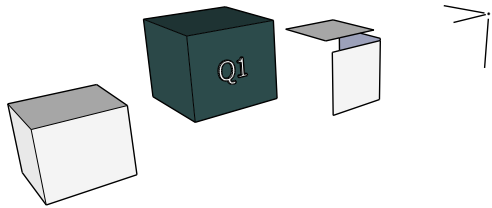}
\caption{ Regions on whose beliefs the message $m_{Q_1\to s_1}(s_1)$ from the cubic region $Q_1$ to the central spin $s_1$ appears. 
On the leftmost diagram, $Q_1$ 
is still represented to help guiding the eye, but it does not belong to $ \partial m_{Q_1\to s_1}$. The cube opposing $Q_1$, 
the three palquettes and the three edges and $s_1$ itself, they all intersect with $Q_1$ only at $s_1$, and therefore are in  $ \partial m_{Q_1\to s_1}$.\label{fig:cube3D_mQ1_s1}}
\end{figure*}

As a consequence also of the second point in the definition, any given message $m_{p\to r}$ is present in the belief  equations of all regions whose intersection with $p$ is exactly $r$:
\[ \partial m_{p\to r} = \{r'\in R | r' \cap p = r\}
\] 
In order for this prescription to be valid, we have to show that it satisfies (\ref{eq:mc_r}).

To keep with our previous 3D example, let us consider the case of the message $m_{Q_1 \to s_1}(s_1)$, going from the cube $Q_1$ (dark) to the central 
spin $s_1$, as shown in Fig. \ref{fig:cube3D_mQ1_s1}. Then the set of regions $ \partial m_{Q_1\to s_1} $ on whose beliefs the message $m_{Q_1\to s_1}(s_1) $
appear are those whose intersection with $Q_1$ is exactly $s_1$, as represented in the diagram (Fig. \ref{fig:cube3D_mQ1_s1}).

\subsubsection{Properties of maximal messages}

\begin{theo}
Equation (\ref{eq:brgl}) defines a valid GBP approximation on any set of regions $R$ defined by the cluster variational method. 
\end{theo}
This theorem is proved in the appendix \ref{ap:mmmmvalid}, based on the fact that the cluster variational method defines a partially ordered set, 
and some properties relating the ancestors of a region and the set of equations in which messages to that region participates. 

Extremal values of $\fcvm$ are  obtained by enforcing eq.~(\ref{eq:dF_dm}). Since all messages appear linearly, differentiating is equivalent to remove the messages from the equations in which they are present. In order to obtain a nicer presentation we can solve instead
\[ m_{r_0\to\gamma} (\underline{x}_{\gamma})  \:\frac{\partial F_R[\{m \}]}{\partial m_{r_0\to\gamma} } =0
\]
which  generates the following set of equations:
\[ \sum_{r\in\partial m_{r_0\to\gamma} } c_r\sum_{\underline{x}_{r}\setminus \underline{x}_{\gamma}} b_r(\xr) = 0
\]
Notice that, as a consequence of equation (\ref{eq:mc_r}), a particular solution to this equation is found when 
each belief involved has the same marginal over the degrees of freedom $\underline{x}_\gamma$.  
Since beliefs are usually interpreted as approximations of the marginals of the joint probability distribution  
(\ref{eq:Px_boltzmann}), we would require them to be consistent with one another. It is 
assuring to see that the consistency indeed is a solution of the extremal equations. In appendix   \ref{ap:beliefs_rotation}
we proof the following
\begin{theo} The extremal points of the approximated variational free energy  $F_R[\{m\}]$ are found at consistent beliefs:
\begin{equation}
\forall r\in R \:\:\: \forall p \in A_r  \quad b_r(\xr) = \sum_{\underline{x}_{p}\setminus \xr}  b_{p}(\underline{x}_{p}) 
\label{eq:consistency}
\end{equation}
\end{theo}

From these equations we can write message passing update rules in different ways. Unfortunately maximal message passing do not 
solves automatically the gauge invariance in the messages. Just as in \ptc\  case, the introduced messages do not define the beliefs 
in an unique way, and many possible messages values may represent the same beliefs. Besides its non optimality as a representation, 
and probably derived efficiency/convergence problems, this invariance might be problematic for average case predictions, 
for instance the prediction of critical temperatures in disordered systems.

Another relevant property of maximal message passing, is that it is hierarchical.
\begin{theo}
Let there be two CVM approximations for a given model. If one of the approximations is contained in the other, meaning that 
all regions in one are present in the other
\[
R_{\text{\tiny CVM}_1} \subset R_{\text{\tiny CVM}_2}
\]
then the beliefs in 
the smaller approximation (the less precise one) are obtained from the larger approximation, by just 
setting to 1 all messages not common to both.
\end{theo}
The proof is immediate since the definition of $D_r$ implies that $D_r^1 \subset D_r^2$. 

As a consequence, if one wants to recover the Bethe approximation from, e.g., the plaquette approximation, we only needs to 
disregard (setting to 1) the plaquette-to-link and plaquette-to-spin messages in the belief and messages passing equations. This property 
is also valid in the case of gauge free maximal messages that we explain next. We recall, however, that two different approximations have 
different counting numbers, and therefore the free energy of the smallest approximation is not obtained by setting to zero
the terms of the larger.

\subsection{Moment matching is gauge free} 
\label{sec:momentmatching}

The general way to create a message passing that is also gauge free starts from recognizing that the relevant quantities to match are not necessarily the belief as in eq.~(\ref{eq:consistency}) but their moments. Next we present the case of Ising variables $s_i = \pm 1$. A more general presentation, 
regarding for instance Potts variables, is left for future work.

Let's go back to the messages. It has been shown in the context of \ptc\  message passing  \cite{GBPGF,zhou_redundant} that when the region graph 
contains loops, i.e. when from a bigger region there are two paths to get to a smaller region in the region graph, then the 
messages are not uniquely determined. In other words, since marginalization is transitive, forcing $b_p \to b_{r_1} \to b_{r_0}$ and 
$b_p \to b_{r_2} \to b_{r_0}$ is redundant. The marginalization $ b_{r_2} \to b_{r_0}$, for instance, automatically follows from the first chain 
of marginalizations and $b_p \to b_{r_2} $. In other words, interpreting the messages as Lagrange multipliers  \cite{yedidia}, the introduction 
of a multiplier to force  $ b_{r_2} \to b_{r_0}$ is unnecessary, and therefore, the set of multipliers is not uniquely determined.

A workaround to this problem has been given previously \cite{GBPGF,zhou_redundant,Pakzad2005Kikuchi} where authors have identified a 
link between gauge invariance and loops in the region graph representation. 
At the end, it all amounts to discovering which are the redundant messages, and remove them from the representation, or set them to 
an arbitrary value, fixing the gauge  \cite{zhou_redundant,Pakzad2005Kikuchi,Aurelien16}. In many case, although the final
objective is clear (destroying the loops), there are many different ways to achieve it, and each one has selected his own way.
Next we explain how to construct gauge-free message passing algorithms from scratch, not by destroying loops, but by restricting
degrees of freedom in the messages. We specify a precise and unique way to do so.

So far, maximal messages were introduced in full generality. Now we will reduce their degrees of freedom
as long as they keep ensuring the consistent marginalization of neighbor regions.  
Proceeding in this way, we do not affect the overall minimization of the CVM free energy.

We will now change perspective and interpret messages not as arbitrary functions enforcing beliefs consistency, but rather as a set of fields enforcing the agreement between the moments in the beliefs. For instance a message to a 2-spin regions, can be rewritten as
\be
m_{p\to1,2}(s_1,s_2) = e^{s_1 s_2 U_{p\to 1,2}  + s_1 u_{p\to 1}  + s_2 u_{p\to 2} }.
\ee
The four values of function $m(s_1,s_2)$ are encoded into the 3 parameters $U,u_1,u_2$ since messages are insensitive to any normalization factor, a consequence of property in eq. (\ref{eq:mc_r}). Let us assume that the fields $u_{p\to 1},u_{p\to 2}$ fix the correct firsts moments 
between the beliefs at region $p$ and at the two-spin region $(1,2)$, while $U$ determines the correlation. In such case, since all parents of region $(1,2)$ are sending messages to it, all those parents will have a first and second moments on variables $s_1$ and $s_2$ that are consistent to that of the belief at $(1,2)$ and therfore are consistent among them. This also means that 
given two ancestors of $(1,2)$, lets say $g,p\in A_{(1,2)}$ such that $g \in A_p$, the messages from $g$ to $p$ do not require fields of the type $u_1$, $u_2$ and $U_{1,2}$ any longer.

An example is handy. Take for instance the 2D square lattice, a small fraction of which is represented here
\begin{center}
\begin{pspicture}(3,3) 
\psline[linewidth=.05cm]{-}(0,0)(-0.3,0)
\psline[linewidth=.05cm]{-}(0,0)(1,0)
\psline[linewidth=.05cm]{-}(1,0)(2,0)
\psline[linewidth=.05cm]{-}(2,0)(3,0)
\psline[linewidth=.05cm]{-}(3,0)(3.3,0)

\psline[linewidth=.05cm]{-}(0,1)(-0.3,1)
\psline[linewidth=.05cm]{-}(0,1)(1,1)
\psline[linewidth=.05cm]{-}(1,1)(2,1)
\psline[linewidth=.05cm]{-}(2,1)(3,1)
\psline[linewidth=.05cm]{-}(3,1)(3.3,1)

\psline[linewidth=.05cm]{-}(0,2)(-0.3,2)
\psline[linewidth=.05cm]{-}(0,2)(1,2)
\psline[linewidth=.05cm]{-}(1,2)(2,2)
\psline[linewidth=.05cm]{-}(2,2)(3,2)
\psline[linewidth=.05cm]{-}(3,2)(3.3,2)
\psline[linewidth=.05cm]{-}(0,3)(-0.3,3)
\psline[linewidth=.05cm]{-}(0,3)(1,3)
\psline[linewidth=.05cm]{-}(1,3)(2,3)
\psline[linewidth=.05cm]{-}(2,3)(3,3)
\psline[linewidth=.05cm]{-}(3,3)(3.3,3)

\psline[linewidth=.05cm]{-}(0,0)(0,-0.2)
\psline[linewidth=.05cm]{o-}(0,0)(0,1)
\psline[linewidth=.05cm]{o-}(0,1)(0,2)
\psline[linewidth=.05cm]{o-o}(0,2)(0,3)
\psline[linewidth=.05cm]{-}(0,3)(0,3.3)

\psline[linewidth=.05cm]{-}(1,0)(1,-0.2)
\psline[linewidth=.05cm]{o-}(1,0)(1,1)
\psline[linewidth=.05cm]{o-}(1,1)(1,2)
\psline[linewidth=.05cm]{o-o}(1,2)(1,3)
\psline[linewidth=.05cm]{-}(1,3)(1,3.3)

\psline[linewidth=.05cm]{-}(2,0)(2,-0.2)
\psline[linewidth=.05cm]{o-}(2,0)(2,1)
\psline[linewidth=.05cm]{o-}(2,1)(2,2)
\psline[linewidth=.05cm]{o-o}(2,2)(2,3)
\psline[linewidth=.05cm]{-}(2,3)(2,3.3)

\psline[linewidth=.05cm]{-}(3,0)(3,-0.2)
\psline[linewidth=.05cm]{o-}(3,0)(3,1)
\psline[linewidth=.05cm]{o-}(3,1)(3,2)
\psline[linewidth=.05cm]{o-o}(3,2)(3,3)
\psline[linewidth=.05cm]{-}(3,3)(3,3.3)

\end{pspicture}
\end{center}
and the square plaquette-CVM approximation. Regions are the plaquettes, the 
links and the spins (variables $s_i =\pm1$) in the system. Any single spin receives messages of the form
\[m_{r\to i} = e^{u_{r\to i} s_i}\:.\]
from the four links and the four plaquettes it belongs to. In the diagram, only the fields coming from one plaquette and two links are shown:

\begin{center}
\begin{pspicture}(4,3.5) 
\psline[linewidth=.05cm]{o-o}(1,1)(2,1)
\psline[linewidth=.05cm]{o-o}(2,1)(2,2)
\psline[linewidth=.05cm]{o-o}(2,2)(1,2)
\psline[linewidth=.05cm]{o-o}(1,2)(1,1)

\psline[linewidth=.05cm]{o-o}(3,1)(3,2)
\psline[linewidth=.05cm]{o-o}(1,3)(2,3)

\pscircle[linewidth=0.05cm](3,3){0.14}

\psline[linewidth=.05cm,doubleline=true]{->}(2.1,1.5)(2.8,1.5)
\psline[linewidth=.05cm,doubleline=true]{->}(1.5,2.1)(1.5,2.8)

\psline[linewidth=.05cm]{->}(2.15,2.15)(2.8,2.8)
\psline[linewidth=.05cm]{->}(2.2,3)(2.8,3)
\psline[linewidth=.05cm]{->}(3,2.2)(3,2.8)
\end{pspicture}
\end{center}

\noindent This ensures that the first moment of the beliefs in the plaquettes and 
the links are consistent with the first moment of the belief at spin $i$, $\<s_i\> = \sum_{s_i} s_i b_i (s_i)$. Therefore, when 
plaquettes are sending messages to links (double arrows in diagram), they no longer need a multiplier (field) $u_{p\to i}$, and the message will only force 
the correlation
\[m_{p\to (ij)} (s_i,s_j) = e^{U_{p\to ij} s_i s_j}\]
In such a way, even though the region graph have loops, the moments are not fixed redundantly, and the message passing is gauge-free. 

In general, a message $m_{p \to r}(\sr)$ has $2^{|r|}-1$ degrees of freedom, where $|r|$ is the number of binary variables (spins) 
in region $r$. There are also $2^{|r|}-1$ non null subsets of $r$, and therefore the same number of moments to describe a distribution 
over $|r|$ variables. The reduction of the degrees of freedom in the messages follow the rule:
\begin{defi}[Moment matching] \label{def:moment_matching}
Message $m_{p\to r}(\sr) $ contains a field $U_{q}$ enforcing the correlation among variables in $q\subseteq r$ as  
\[m_{p\to r}(\sr)  = e^{\cdots+U^{p\to r}_q \prod_{i\in q} s_i + \cdots}
\] 
if and only if  $r$ is the smallest region among all those containing the variables in $q$.
\end{defi}

The smallest region containing $q$ is uniquely determined in the CVM construction thanks to the following properties: (i) the partial order defined by inclusion relations and (ii) the closure of the set of CVM regions under intersections of sets (the proof is easy and left to the reader).

Extreme values of the approximated free energy $\fcvm$ can now be obtained by differentiating directly with respect to the fields $U_{q}$ that define the messages:
\[\frac{\partial \fcvm[\{U \}]}{\partial U^{p\to \gamma}_q } = 0
\]
which  generates the following set of equations:
\[ \sum_{r\in\partial m_{p \to \gamma} }  c_r \left<\prod_{i\in q} s_i \right>_{b_r} = 0
\]
and 
\[ \left<\cdots \right>_{b_r} = \sum_{\sr} \cdots b_r(\underline s_r)
\]
Obviously, a particular solution is found when all distributions share the same moments over common
degrees of freedom, since equation (\ref{eq:mc_r}) holds. In such case, all beliefs are also consistent 
with inner regions. It remains to show that this is in fact the only 
solution, which is the argument of the following
\begin{theo}\label{theo:mmmmconsistent}
 Maximal messages with moment matching fields ensures the consistency of beliefs.
\end{theo}

The previous theorem states that the moment matching fields are enough to guarantee consistency. The next one completes our task 
by stating that indeed we need all of these fields to do so.
\begin{theo}\label{theo:mmmmgauge}
 Maximal messages with moment matching fields is gauge free.
\end{theo}

Both theorems are proved in appendix \ref{ap:mmmmgauge}




\section{Plaquette-CVM for Ising 2D}
\label{sec:isingd}

Let us start by a simple case. Ising ferromagnet, in the absence of external fields are defined
by the Hamiltonian
\[ \Ham (\alls)  = - J \sum_{\<i,j\>} s_i  s_j \]
where $\<i,j\>$ defines nearby spins, and is given by the topology in which the system is embedded, and the degrees of freedom  
are $s_i =\pm 1$. The interaction constant $J$ is normally set to $J=1$.

Though the 2D case of this model has been exactly solved  \cite{onsager}, we still can try our approximation on it, 
before moving to the unsolved higher dimensions. The first approximation beyond mean field and Bethe, is the one containing 
all square plaquettes (the basic cell)  as maximal regions. The cluster variation method then prescribe a free energy in 
terms of Plaquettes, Links and Spins regions \cite{GBPGF}, with counting numbers $c_P =1, c_L=-1, c_i = 1$ respectively.

The gauge free 4M-CVM is then written in terms of messages going from plaquettes to the  links and spins interior to it. Beliefs are defined as follows 
\begin{eqnarray}
b_P(\sfour) & =& \frac1 {z_P} e^{-\B E_4(\sfour)} \prod_{L\in P}\prod_{\substack{ P'\supset L\\ P'\neq P}} m_{P'\to L}(s,s') \nonumber \\ 
&& \phantom{e^{-\B E_4(\sfour)}} \prod_{s\in P}\prod_{\substack{ P'\cap P = s}} \textcolor[rgb]{1,0,0}{ m_{P'\to s} (s)}\label{eq:beliefP} \\ 
&& \phantom{e^{-\B E_4(\sfour)}} \prod_{\substack{L\notin P \\ L\cap P\neq \varnothing}}  \textcolor[rgb]{0,0,1}{ m_{L\to s=L\cap P}}(s) \nonumber
\end{eqnarray} 
\begin{eqnarray}
b_L(s,s') & =& \frac1 {z_L}e^{-\B E_2(s,s')} \prod_{\substack{ P\supset L}} m_{P\to L}(s,s') \label{eq:beliefL} \\ 
&& \phantom{e^{-\B E_2(s,s')}} \prod_{s\in L} \prod_{\substack{ P' \cap L = s}} \textcolor[rgb]{1,0,0}{ m_{P'\to s}(s)} \nonumber\\
&& \phantom{e^{-\B E_2(s,s')}} \prod_{\substack{L' \neq L \\ L'\cap L\neq \varnothing}} \textcolor[rgb]{0,0,1}{ m_{L'\to s=L'\cap L}(s)} \nonumber
\end{eqnarray} 
\begin{eqnarray}
b_s(s) & =& \frac1 {z_s} e^{-\B E_1(s)} \prod_{P \supset s}\textcolor[rgb]{1,0,0}{ m_{P\to s}(s)} \label{eq:beliefs}\\
&& \phantom{e^{-\B E_2(s,s')}} \prod_{L \supset s}  \textcolor[rgb]{0,0,1}{m_{L\to s}(s) }\nonumber
\end{eqnarray} 

Graphically, the beliefs of each region are given by 

\begin{pspicture}(7,3) 
\psline[linewidth=.05cm]{o-o}(1,1)(2,1)
\psline[linewidth=.05cm]{o-o}(2,1)(2,2)
\psline[linewidth=.05cm]{o-o}(2,2)(1,2)
\psline[linewidth=.05cm]{o-o}(1,2)(1,1)
\psline[linewidth=.05cm]{o-o}(4,1)(4,2)
\pscircle[linewidth=0.05cm](6,2){0.14}

\psline[linewidth=.05cm,doubleline=true]{->}(0.2,1.5)(0.9,1.5)
\psline[linewidth=.05cm,doubleline=true]{->}(2.8,1.5)(2.1,1.5)
\psline[linewidth=.05cm,doubleline=true]{->}(1.5,2.8)(1.5,2.1)
\psline[linewidth=.05cm,doubleline=true]{->}(1.5,0.2)(1.5,0.9)

\psline[linewidth=.05cm,linecolor=red]{->}(0.3,2.6)(0.9,2.1)
\psline[linewidth=.05cm,linecolor=blue]{->}(0.3,2.0)(0.9,2.0)
\psline[linewidth=.05cm,linecolor=blue]{->}(1,2.6)(1,2.1)

\psline[linewidth=.05cm,linecolor=red]{->}(2.7,2.6)(2.1,2.1)
\psline[linewidth=.05cm,linecolor=blue]{->}(2.7,2)(2.1,2)
\psline[linewidth=.05cm,linecolor=blue]{->}(2,2.6)(2,2.1)

\psline[linewidth=.05cm,linecolor=red]{->}(0.3,0.4)(0.9,0.9)
\psline[linewidth=.05cm,linecolor=blue]{->}(0.3,1)(0.9,1)
\psline[linewidth=.05cm,linecolor=blue]{->}(1,0.4)(1,0.9)

\psline[linewidth=.05cm,linecolor=red]{->}(2.7,0.4)(2.1,0.9)
\psline[linewidth=.05cm,linecolor=blue]{->}(2.7,1)(2.1,1)
\psline[linewidth=.05cm,linecolor=blue]{->}(2,0.4)(2,0.9)

\psline[linewidth=.05cm,doubleline=true]{->}(3.2,1.5)(3.9,1.5)
\psline[linewidth=.05cm,doubleline=true]{->}(4.8,1.5)(4.1,1.5)

\psline[linewidth=.05cm,linecolor=red]{->}(3.3,2.6)(3.9,2.1)
\psline[linewidth=.05cm,linecolor=blue]{->}(4.7,2)(4.1,2)
\psline[linewidth=.05cm,linecolor=blue]{->}(3.3,2)(3.9,2)
\psline[linewidth=.05cm,linecolor=blue]{->}(4,2.6)(4,2.1)

\psline[linewidth=.05cm,linecolor=red]{->}(4.7,2.6)(4.1,2.1)
\psline[linewidth=.05cm,linecolor=red]{->}(3.3,0.4)(3.9,0.9)
\psline[linewidth=.05cm,linecolor=red]{->}(4.7,0.4)(4.1,0.9)
\psline[linewidth=.05cm,linecolor=blue]{->}(4.7,1)(4.1,1)
\psline[linewidth=.05cm,linecolor=blue]{->}(3.3,1)(3.9,1)
\psline[linewidth=.05cm,linecolor=blue]{->}(4,0.3)(4,0.9)

\psline[linewidth=.05cm,linecolor=red]{->}(5.3,2.6)(5.9,2.1)
\psline[linewidth=.05cm,linecolor=red]{->}(6.7,2.6)(6.1,2.1)
\psline[linewidth=.05cm,linecolor=red]{->}(5.3,1.4)(5.9,1.9)
\psline[linewidth=.05cm,linecolor=red]{->}(6.7,1.4)(6.1,1.9)
\psline[linewidth=.05cm,linecolor=blue]{->}(6.7,2)(6.1,2)
\psline[linewidth=.05cm,linecolor=blue]{->}(6,1.3)(6,1.9)
\psline[linewidth=.05cm,linecolor=blue]{->}(5.3,2)(5.9,2)
\psline[linewidth=.05cm,linecolor=blue]{->}(6,2.6)(6,2.1)
\end{pspicture}
where double arrows represent messages to links $m_{P\to L}(s,s')$, oblique arrows messages from plaquettes to spins $\textcolor[rgb]{1,0,0}{m_{P\to s}(s)}$ and remaining arrows messages from links to spins $\textcolor[rgb]{0,0,1}{m_{L\to s}(s)}$. Colors have been added  (online version)  to help identify each arrow with its corresponding term.

\subsubsection{Message passing}
Message passing equations can be obtained in two different but equivalent ways:
\begin{itemize}
 \item Old way: by imposing the consistency among beliefs, in this case some of the following:
 \begin{eqnarray*}
 b_L(s_1,s_2) &=& \sum_{s_3,s_4} b_P(s_1,s_2,s_3,s_4) \\   
 b_s(s_1) &=& \sum_{s_2,s_3,s_4} b_P(s_1,s_2,s_3,s_4) \\   
 b_s(s_1) &=& \sum_{s_2} b_L(s_1,s_2) \\   
 \end{eqnarray*}
 \item New way: by imposing consistency among the moments of the distributions:
 \begin{eqnarray*}
 \sum_{s_1,s_2} s_1 s_2 \: b_L(s_1,s_2) &=&  \sum_{s_1,s_2,s_3,s_4} s_1 s_2 \: b_P(s_1,s_2,s_3,s_4) \\   
 \sum_{s_1} s_1 \: b_s(s_1) &=& \sum_{s_1, s_2,s_3,s_4} s_1 \: b_P(s_1,s_2,s_3,s_4) \\   
 \sum_{s_1}s_1\: b(s_1) &=& \sum_{s_2} s_1 \: b_L(s_1,s_2) \\   
 \end{eqnarray*}
\end{itemize}
Furthermore, as can be easily checked not all three equations in the old way are independent:
the third equation is consequence of the first two. This is the very reason why we reduced the amount of fields. 
In the new way there is only three values being fixed, and they are all independent. 
Both ways, however, produce the same update equations (message passing)
independently of whether the messages has been reduced to be gauge fixed, or are in full generality.

For instance, forcing any link belief $b_L(s_1,s_2)$ to marginalize onto one of its spins
results in the following equation:
\begin{multline}
\sum_{s_2} \:e^{\B J s_1 s_2 } m_{P_1\to L}(s_1,s_2) m_{P_2\to L}(s_1,s_2)\\
m_{P_3\to s_2}(s_2) m_{P_4\to s_2}(s_2)  \prod_{\substack{L' \supset s_2\\ L'\neq L}} m_{L'\to s_2}(s_2)\\
\propto m_{L \to s_1}(s_1)\prod_{P \supset L} m_{P \to s_1}(s_1)
\label{eq:Ltosmm}
\end{multline}
where we put a sign of proportionality $\propto$ instead of equality since messages are undefined by a multiplicative constant.
These equations can be derived graphically using the representations of the beliefs and messages introduced above.
The rules are quite simple. Interactions are represented by the rods, degrees of freedom by the circles, and messages by the arrows. If an 
interaction or a message appears in both sides of the equations, can be cancel out. The degrees of freedom over which the marginalization 
is carried appear as full black circles. For instance, equation (\ref{eq:Ltosmm}) is represented as in figure \ref{fig:link-to-spin}.

\begin{figure}[h]
\begin{center}
\begin{pspicture}(7,3)

\psline[linewidth=.05cm]{*-o}(4,1)(4,2)
\pscircle[linewidth=0.05cm](6,2){0.12}

\psline[linewidth=.03cm]{-}(5,1.9)(5.3,1.9)
\psline[linewidth=.03cm]{-}(5,1.8)(5.3,1.8)

\psline[linewidth=.05cm,doubleline=true]{->}(3.2,1.5)(3.9,1.5)
\psline[linewidth=.05cm,doubleline=true]{->}(4.8,1.5)(4.1,1.5)

\psline[linewidth=.05cm]{->}(3.3,0.4)(3.9,0.9)
\psline[linewidth=.05cm]{->}(4.7,0.4)(4.1,0.9)

\psline[linewidth=.05cm]{->}(5.3,1.4)(5.9,1.9)
\psline[linewidth=.05cm]{->}(6.7,1.4)(6.1,1.9)

\psline[linewidth=.05cm]{->}(3.3,0.4)(3.9,0.9)
\psline[linewidth=.05cm]{->}(4.7,0.4)(4.1,0.9)
\psline[linewidth=.05cm]{->}(4.7,1)(4.1,1)
\psline[linewidth=.05cm]{->}(3.3,1)(3.9,1)
\psline[linewidth=.05cm]{->}(4,0.3)(4,0.9)

\psline[linewidth=.05cm]{->}(6,1.3)(6,1.9)
\end{pspicture}
\end{center}
\caption{Consistency equation between link beliefs and spin beliefs. Mathematically it corresponds to 
the first two equations in (\ref{eq:Uuu}) for $d=2$.} \label{fig:link-to-spin}
\end{figure}
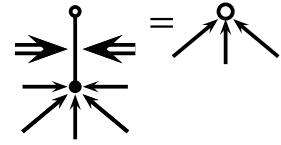

The plaquette to link marginalization produces consistency relation between messages as shown in figure \ref{fig:plaqu-to-link}.

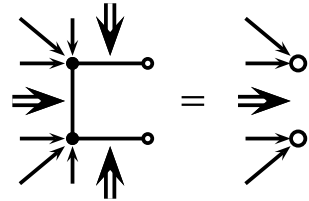
\begin{figure}[h]
\begin{center}
\begin{pspicture}(7,3) 
\psline[linewidth=.05cm]{*-o}(1,1)(2,1)
\psline[linewidth=.05cm]{o-o}(2,2)(1,2)
\psline[linewidth=.05cm]{*-*}(1,2)(1,1)

\psline[linewidth=.05cm,doubleline=true]{->}(0.2,1.5)(0.9,1.5)
\psline[linewidth=.05cm,doubleline=true]{->}(1.5,2.8)(1.5,2.1)
\psline[linewidth=.05cm,doubleline=true]{->}(1.5,0.2)(1.5,0.9)

\psline[linewidth=.03cm]{-}(2.45,1.55)(2.75,1.55)
\psline[linewidth=.03cm]{-}(2.45,1.45)(2.75,1.45)

\psline[linewidth=.05cm]{->}(0.3,2.6)(0.9,2.1)
\psline[linewidth=.05cm]{->}(0.3,2.0)(0.9,2.0)
\psline[linewidth=.05cm]{->}(1,2.6)(1,2.1)

\psline[linewidth=.05cm]{->}(0.3,0.4)(0.9,0.9)
\psline[linewidth=.05cm]{->}(0.3,1)(0.9,1)
\psline[linewidth=.05cm]{->}(1,0.4)(1,0.9)

\pscircle[linewidth=0.05cm](4,2){0.12}
\pscircle[linewidth=0.05cm](4,1){0.12}
\psline[linewidth=.05cm,doubleline=true]{->}(3.2,1.5)(3.9,1.5)

\psline[linewidth=.05cm]{->}(3.3,2.6)(3.9,2.1)
\psline[linewidth=.05cm]{->}(3.3,2)(3.9,2)

\psline[linewidth=.05cm]{->}(3.3,0.4)(3.9,0.9)
\psline[linewidth=.05cm]{->}(3.3,1)(3.9,1)
\end{pspicture}
\end{center}
\caption{Consistency equation between plaquette beliefs and link beliefs. Mathematically it corresponds to 
the first two equations in (\ref{eq:Uuu}) for $d=2$.
}\label{fig:plaqu-to-link}
\end{figure}

As can be seen, in either cases (link to spin and plaquette to spin) the messages in the right hand side do not appear isolated. 
Consistency equations force the product of messages. We could have used plaquette to spin marginalization as well, and the situation 
still would be similar. In such cases, it is left to the programmer to decide which iterative updating rule she wishes to implement 
to solve the consistency equations in a message passing way. She could, for instance, use the link to spin equation to update both 
plaquette to spin messages in a symmetric way, and then use the plaquette to spin equation to update the plaquette to link message. 
Let us emphasize that this freedom on the implemenation of message passing equations remains even when the gauge is fixed,
just as any fixed point equation can be written in infinte many ways. The gauge fixed property refers to the unicity of fields values 
at a given fixed point, not to the strategies to find them.

If messages are considered in full generality, then we have a redundant description
\begin{eqnarray*}
m_{P \to L} (s_1,s_2) &=& e^{ U_{P \to L} s_1 s_2 + u_{P \to 1} s_1 + u_{P \to 2} s_2} \\
m_{L \to 1} (s_1) &=& e^{ u_{L \to 1} s_1} \\
m_{L \to 2} (s_2) &=& e^{ u_{L \to 2} s_2} 
\end{eqnarray*}
leading to a gauge invariance transformation involving $u$ messages \cite{GBPGF}. On the contrary, using the gauge-free moment matching prescription previously defined in Def.~\ref{def:moment_matching}, messages are
\begin{eqnarray*}
m_{P \to L} (s_1,s_2)&=& e^{ U_{P\to L} s_1 s_2} \\
m_{P \to 1} (s_1) &=& e^{u_{P\to 1} s_1} \\
m_{L \to 1} (s_1) &=& e^{u_{L\to 1} s_1}
\end{eqnarray*}
Details of the update equations and an example on 2D single instance is given next.


\subsection{2D Single instance implementation}


The self consistent message passing equations can be written as $ \textcolor[rgb]{0,0,1}{u_L} = \hatu_L(\B,J,u_L,U,u_P)$ and 
$U = \hatU(\B,J,u_L,U,u_P) $,  $\textcolor[rgb]{1,0,0}{u_P} = \hatup(\B,J,u_L,U,u_P) $, where 
\begin{eqnarray}
\hatu_L& =&   \frac 1 2 \log(K(1)/K(-1)) - u_{P_a} - u_{P_b}  \nonumber \\
\hatup& = & \frac 1  4  \log\left(\frac{K(1,1)K(1,-1)}{K(-1,1)K(-1,-1)}\right) - u_{L} \nonumber \\
\hatU & = & \frac 1  4 \log\left(\frac{K(1,1)K(-1,-1)}{K(1,-1)K(-1,1)}\right)  \label{eq:uuU} 
\end{eqnarray}
The $K(\cdot)$ terms are partial traces over the spins in the plaquete and link, given by:
\begin{widetext}
\begin{eqnarray}
K(s_1) &=& \sum_{s_2} e^{    (\B J_{12} + U_1 + U_2)) s_1 s_2 + ( u_{P_1} +  u_{P_2} +  u_{L1} +u_{L2} +u_{L3}) s_2 } \nonumber \\
 K(s_1,s_2)&=& \sum_{s_3,s_4} \exp \left[ s_2s_3 (\B J_{23}+ U_{23})+s_1 s_4 (\B J_{14}+U_{14})+ s_3 s_4 (\B J_{34}+U_{34})  +  \right. \nonumber\\
&&\left. \phantom{ \sum_{s_3,s_4} } +\big( u_{P_4} +  u_{L_1\to4} +  u_{L_2\to4}\big) s_4  +  \big(u_{P_3} +  u_{L_1\to3} +  u_{L_2\to3}\big) s_3  \right]  \nonumber
 \end{eqnarray}
\end{widetext}
in correspondence with the fields in the left hand sides of diagrams \ref{fig:link-to-spin} and \ref{fig:plaqu-to-link}. 

\begin{figure}[!htb]
  \includegraphics[width=0.33\textwidth,angle=270]{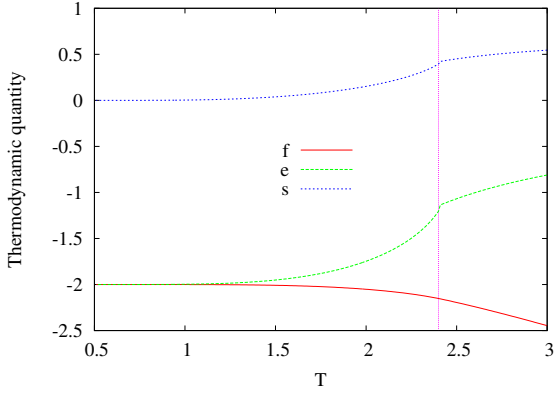}
  \caption{ Estimated intensive thermodynamic quantities for the Ising 2D 
  model using gauge-free message passing.\label{fig:thermo_ferro}}
\end{figure}

The implementation of the message passing is carried by randomly selecting a plaquette (or link) and updating their fields as 
prescribed by equations (\ref{eq:uuU}).
In 2D Ising model we obtain the expected results (see Fig. \ref{fig:thermo_ferro}). 
Above the approximation critical temperature (not exact) $T_c =1/\beta_c \simeq 2.43$ all fields acting on single spins are 
zero $u_{L \to i} = u_{P\to i}=0$, and the system is in a paramagnetic phase 
with zero global magnetization. In this range, the only non zero field is 
the correlation field $U_{P\to L}$ (a detailed studied of this phase is in \cite{dual}). Below $T_c$, the system is in a 
ferromagnetic phase, with non zero fields over spins and local as well 
as global magnetization.

Next we show how to generalize this method to compute the critical temperature of the Ising model in general dimension, under the plaquette-CVM approximation.

\subsection{Critical temperature for Ising d-dimensional}

Let us focus on the case of the plaquette CVM approximation in the general d-dimensional Ising
model on the hypercubic lattice. This case includes  the model of the previous section.

We will show how to obtain analytic expression for the critical temperature of the ferromagnetic model in this approximation at all 
dimensions, and furthermore, we will show that the asymptotic behavior is correctly until the third order in $1/d$, therefore being equivalent to the loop corrections of  \cite{loop_corrected_tom}.

The plaquette approximation is the one that uses plaquettes as the biggest regions. In such case,
the counting numbers of the plaquettes are always $c_P=1$. Every link belongs to $2(d-1)$ plaquettes, and 
therefore its counting number is $c_L = 1-2(d-1)$. Every spin belongs to $2d$ links and $2d(d-1)$ plaquettes 
and have counting number $c_s = 1-2d(d-1) -2d [1-2(d-1)] =1-2 d(2-d)$. Beliefs, therefore, have the following 
schematic representation, where, as usual, double arrows are messages from plaquette to link, oblique arrows 
from plaquette to spin and remaining (vertical and horizontal) arrows are from links to spins.

\begin{pspicture}(7,4) 
\psline[linewidth=.05cm]{o-o}(1,1)(2,1)
\psline[linewidth=.05cm]{o-o}(2,1)(2,2)
\psline[linewidth=.05cm]{o-o}(2,2)(1,2)
\psline[linewidth=.05cm]{o-o}(1,2)(1,1)
\psline[linewidth=.05cm]{o-o}(4,1)(4,2)
\pscircle[linewidth=0.05cm](6,2){0.14}


\psline[linewidth=.05cm,doubleline=true]{->}(1.5,2.8)(1.5,2.1)
\rput(1.5,3.0){$2(d-1)-1$}


\psline[linewidth=.05cm]{->}(1,0.3)(1,0.9)
\rput(1.0,0.15){$2d-2$}

\psline[linewidth=.05cm]{->}(2.7,0.4)(2.1,0.9)
\rput(3.7,0.15){$2d(d-1)-4d+5$}

\psline[linewidth=.05cm,doubleline=true]{->}(3.2,1.5)(3.9,1.5)
\rput(3.2,1.1){$2(d-1)$}
\psline[linewidth=.05cm]{->}(4.7,1)(4.1,1)
\rput(4.8,0.7){$2d-1$}

\psline[linewidth=.05cm]{->}(3.3,2.6)(3.9,2.1)
\rput(4.0,2.9){$2(d-1)^2$}

\psline[linewidth=.05cm]{->}(6.7,2.6)(6.1,2.1)
\rput(5.7,2.7){$2d (d-1)$}
\psline[linewidth=.05cm]{->}(6,1.3)(6,1.9)
\rput(6.2,1){$2d$}

\end{pspicture}

\noindent For clarity,  only one type of message of each type is represented in each region together with the number of such messages that enter in 
the belief equation of that region. However, the reader should keep in mind that, for instance, there are $2d(d-1)-4d+5$ plaquette-to-spin fields entering
at very corner of the represented plaquette.

In general, consistency equations for messages keep the same structure represented graphically in the previous section, but only the 
amount of messages entering every region changes. Exploiting the isotropy of the model, we can look for fixed points in which all messages 
are the same. In other words, we will assume that all link to spin messages are characterized by a unique field $u_L$, while plaquette to 
link messages by the field $U$   and plaquette to spin messages by $u_P$.

Graphically, the updating equations for the messages have the following representation:

\begin{pspicture}(7,3) 

\psline[linewidth=.05cm]{*-o}(4,1)(4,2)
\pscircle[linewidth=0.05cm](6,2){0.14}

\psline[linewidth=.03cm]{-}(4.6,1.9)(4.9,1.9)
\psline[linewidth=.03cm]{-}(4.6,1.8)(4.9,1.8)

\psline[linewidth=.05cm,doubleline=true]{->}(3.2,1.5)(3.9,1.5)
\rput(2.3,1.5){$2(d-1)$}

\psline[linewidth=.05cm]{->}(3.3,0.4)(3.9,0.9)
\rput(2.4,0.4){$2(d-1)^2$}
\psline[linewidth=.05cm]{->}(4,0.4)(4.0,0.9)
\rput(4.0,0.2){$2d-1$}

\psline[linewidth=.05cm]{->}(6,1.3)(6,1.9)
\rput(6,1.1){$1$}
\psline[linewidth=.05cm]{->}(5.3,1.4)(5.9,1.9)
\rput(5.0,1.1){$2(d-1)$}

\end{pspicture}

\begin{pspicture}(7,4) 
\psline[linewidth=.05cm]{*-o}(3,1)(4,1)
\psline[linewidth=.05cm]{o-*}(4,2)(3,2)
\psline[linewidth=.05cm]{*-*}(3,2)(3,1)
\rput(1.1,1.5){$2(d-1)-1$}

\psline[linewidth=.05cm]{->}(5.3,1)(5.9,1)
\psline[linewidth=.05cm]{->}(5.3,2)(5.9,2)

\psline[linewidth=.05cm]{->}(2.3,1)(2.9,1)
\psline[linewidth=.05cm]{->}(2.3,2)(2.9,2)
\rput(1.6,1){$2d-2$}

\pscircle[linewidth=0.03cm](6,1){0.1}
\pscircle[linewidth=0.03cm](6,2){0.1}

\psline[linewidth=.03cm]{-}(4.5,1.55)(4.8,1.55)
\psline[linewidth=.03cm]{-}(4.5,1.45)(4.8,1.45)
\psline[linewidth=.05cm,doubleline=true]{->}(2.2,1.5)(2.9,1.5)
\psline[linewidth=.05cm,doubleline=true]{->}(3.5,2.8)(3.5,2.1)
\psline[linewidth=.05cm,doubleline=true]{->}(3.5,0.2)(3.5,0.9)
\rput(3.5,0){$2(d-1)-1$}
\psline[linewidth=.05cm]{->}(2.3,2.6)(2.9,2.1)
\rput(1.8,2.9){$2d(d-1)-4d+5$}
\psline[linewidth=.05cm]{->}(2.3,0.4)(2.9,0.9)
\rput(5.0,1.5){$1$}
\rput(5.0,2.0){$1$}
\psline[linewidth=.05cm,doubleline=true]{->}(5.2,1.5)(5.9,1.5)
\rput(5.0,2.9){$2d-3$}
\psline[linewidth=.05cm]{->}(5.3,2.6)(5.9,2.1)
\psline[linewidth=.05cm]{->}(5.3,0.4)(5.9,0.9)
\end{pspicture}

\noindent Exploiting the symmetry of the problem, we gain some clarity by showing only the multiplicity of one representative of each type of fields. 
In the second diagram, fields not having a number by its side, have the same multiplicity that is represented for its equivalent by a reflection 
along the horizontal axis.

Let us define 
\begin{widetext}
\begin{eqnarray}
K(s_1) &=& \sum_{s_2} e^{    (\B J + (2 (d - 1) U)) s_1 s_2 + (2 (d - 1)^2 u_P + (2 d - 1) u_L) s_2 } \nonumber \\
 K(s_1,s_2)&=& \sum_{s_3,s_4} \exp \left[ s_2s_3 (\B J+(2d-3)U)+s_1 s_4 (\B J+(2d-3)U)+ s_3 s_4 (\B J+(2d-3)U)  +  \right. \nonumber\\
&&\left. \phantom{ \sum_{s_3,s_4} } +\big((2 d (d - 1) - 4 d + 5) u_P + (2d-2 ) u_L\big) s_4  +  \big((2 d (d - 1) - 4 d + 5) u_P +  (2d-2 ) u_L\big) s_3  \right]  \nonumber
 \end{eqnarray}
\end{widetext}
In terms of this, the self consistent equations can be written as 
\begin{eqnarray}
U &=& \hatU(\B,J,u_L,U,u_P) \nonumber\\
  &=& \frac 1  4 \log\left(\frac{K(1,1)K(-1,-1)}{K(1,-1)K(-1,1)}\right) \nonumber \\
u_P &=& \hatup(\B,J,u_L,U,u_P) \label{eq:Uuu}\\
  &=& \frac 1 {2d-3} \left[\frac 1  4  \log\left(\frac{K(1,1)K(1,-1)}{K(-1,1)K(-1,-1)}\right) - u_L\right] \nonumber\\
u_L &=& \hatu_L(\B,J,u_L,U,u_P) \nonumber\\
  &=& \frac 1 2 \log(K(1)/K(-1)) - 2(d-1) u_P\:. \nonumber 
\end{eqnarray}
The solution to this set of equations is to be found numerically in general. A simpler case is that of the high temperatures, in which we 
suppose a paramagnetic phase characterized by $u_L=u_P=0$ and $U\neq 0$. In such case the equation $U= \hatU$ becomes the simpler
\begin{equation}
U= \arctanh \Big[ \Big(\tanh\big( (2d-3) U + \B J\big) \Big)^3\Big] \:.\label{eq:dual}
\end{equation}
This corresponds to the case treated in  \cite{dual}.

Moreover the paramagnetic solution is the starting point to obtain the critical temperature of the system as the instability of the 
paramagnetic solution. Taking
\[ \matstab(\B)  = \left(\begin{array}{cc} 
1-\frac {\partial \hatu_L } {\partial u_L} &\frac {\partial \hatu_L } {\partial  u_P} \vspace{0.3cm} \\
\frac {\partial \hatup } {\partial u_L} &1-\frac {\partial \hatup } {\partial u_P}
\end{array} \right)_{u_L=0,U=\hatU,u_P=0}
\]
a continuous instability appears at the point in which $\matstab(\B) $ is singular, and therefore the critical temperature is defined as 
\[\det \matstab (\B_c)  = 0 \:.\]
Note that this is not fully analytical at this point, since the numerical solution of (\ref{eq:dual}) is still needed. However, after some transformations we obtain an analytic expression for the critical temperature at all dimensions $d>2$:
\begin{equation*}
\B_{\text{\tiny CVM}} = \frac 1 2 \log \left[ \displaystyle \left(\frac{d-2}{d}\right)^{d-2} \left(\frac {2d-1}{2d-3}\right)^{2d-3} \right] \:.
\end{equation*}
In the $d=2$ case, the solution is also analytical but given by
\[
\B_{\text{\tiny CVM}}(d=2) = \frac 1{2 } \log\left( \frac{5+\sqrt{17}}{4} \right) \:.
\]
This prediction can be compared with known results. In table~\ref{tab:loop} we show the best estimate of the true $\beta_c$ on a regular hypercubic lattice with $2\le d\le6$, together with the estimate from plaquette CVM, that from the Bethe approximation, where $\B_{\text{\scriptsize{Bethe}}} = \atanh[(2d-1)^{-1}]$, and the one from Bethe with loop corrections due to Rizzo and Montanari \cite{loop_corrected_tom}. In the latter approximation the critical temperature can be computed only if $d>2$.

\begin{table}
\begin{tabular}{l|l|l|l|l}
 &  & plaquette & loop corr &  \\
$d$ & \ \ true $\beta_c$ & CVM & Bethe & Bethe  \\
\hline
2 & 0.440687 (exact)                           & 0.412258 & ---      & 0.346574 \\
3 & 0.221654(6)  \cite{barber_ising_3d}        & 0.216932 & 0.238520 & 0.202733 \\
4 & 0.14966(3)  \cite{gaunt1979susceptibility} & 0.148033 & 0.151650 & 0.143841 \\
5 & 0.11388(3)  \cite{parisi_lorenzo_5d}       & 0.113362 & 0.114356 & 0.111572 \\
6 & ---                                        & 0.092088 & 0.092446 & 0.091161 \\
\end{tabular}
\caption{Inverse critical temperatures of the Ising model on a regular hypercubic lattice in $d$ dimensions. In the second column we report the best estimate for the true $\beta_c$, while the other columns contain the inverse critical temperatures in 3 different mean field approximations: the plaquette CVM discussed in this work, the loop corrected Bethe of Ref.~\cite{loop_corrected_tom}, that can be computed only for $d>2$, and the standard Bethe approximation.\label{tab:loop}}
\end{table}

In the large $d$ limit, the plaquette-CVM critical temperature is correct up to the second order in the $1/d$ expansion, exactly as the loop corrected Bethe approximation \cite{loop_corrected_tom}
\[
\frac{1}{2 d \B_{\text{\tiny CVM}}} = \rlap{$\overbrace{ \phantom{1-\frac 1 {2d} }}^{\text{Bethe}}$}  \underbrace{1-\frac 1 {2d} -\frac 1 {3d^2}}_{\text{Loop corr. Bethe}} - \frac 5 {12 d^3} +\ldots
\]
while the standard Bethe approximation is correct only up to the $O(1/d)$ term.

We find this result very interesting, because it improves over the Bethe approximation by one order of magnitude in the $1/d$ 
expansion, while still providing a very accurate critical temperature at $d=2$.
On the contrary the loop corrected Bethe approximation is divergent in $d=2$, and this makes the 
present 4M-CVM much more useful for the study of low dimensional systems.






\section{Conclusions}

We have shown how to create gauge-free message passing implementations of the cluster variational approximations for general models of spin-like variables. 
To do so we presented a new 
way of introducing the messages in the CVM that differs from standard parent-to-child messages in that messages are sent to a region
from all its ancestors, and not only by its direct parents. While previous attempts to fix the guage invariance in GBP equations \cite{zhou_redundant,Aurelien16,Pakzad2005Kikuchi} relied on the idea of removing some selected messages from the equations, our approach 
increases the number of such messages, but with a restriction on their degrees of freedom.


This systematic restriction of messages degrees  of freedom automatically produces  gauge-free variational approximations, such that
there is a one-to-one correspondence between free energy minima, and the values of the fields that define the messages. 
Furthermore, we put emphasis in a new interpretation of the fields involved in the message passing as imposing consistency between moments of the 
local distributions (beliefs) rather 
than the usual interpretation of messages forcing consistent beliefs marginalization. 
We called the resulting method, maximal messages with moment matching (4M-CVM).

The approach includes the Bethe approximation as the starting point, and improves it when larger regions are taken in consideration.
We showed  that the method  produces sensible analytical results for the plaquette approximation of the critical temperature of the 
Ising ferromagnet in general dimensions, that correctly accounts for the next leading order in the high dimensional expansions, just as the more complicated 
loop calculus does \cite{loop_corrected_tom}.


\section{Acknowledgement} 

The authors want to thank the hospitality of Erik Aurell and his group at Nordita, during the last part of the elaboration of this paper. We also thank 
Daniel de la Regata for support with 3D graphics. This work has been supported in part by the ERC under the European Union’s 7th Framework Programme Grant Agreement 307087-SPARCS.

\bibliography{bibliografia}

\begin{thebibliography}{29}%
\makeatletter
\providecommand \@ifxundefined [1]{%
 \@ifx{#1\undefined}
}%
\providecommand \@ifnum [1]{%
 \ifnum #1\expandafter \@firstoftwo
 \else \expandafter \@secondoftwo
 \fi
}%
\providecommand \@ifx [1]{%
 \ifx #1\expandafter \@firstoftwo
 \else \expandafter \@secondoftwo
 \fi
}%
\providecommand \natexlab [1]{#1}%
\providecommand \enquote  [1]{``#1''}%
\providecommand \bibnamefont  [1]{#1}%
\providecommand \bibfnamefont [1]{#1}%
\providecommand \citenamefont [1]{#1}%
\providecommand \href@noop [0]{\@secondoftwo}%
\providecommand \href [0]{\begingroup \@sanitize@url \@href}%
\providecommand \@href[1]{\@@startlink{#1}\@@href}%
\providecommand \@@href[1]{\endgroup#1\@@endlink}%
\providecommand \@sanitize@url [0]{\catcode `\\12\catcode `\$12\catcode
  `\&12\catcode `\#12\catcode `\^12\catcode `\_12\catcode `\%12\relax}%
\providecommand \@@startlink[1]{}%
\providecommand \@@endlink[0]{}%
\providecommand \url  [0]{\begingroup\@sanitize@url \@url }%
\providecommand \@url [1]{\endgroup\@href {#1}{\urlprefix }}%
\providecommand \urlprefix  [0]{URL }%
\providecommand \Eprint [0]{\href }%
\providecommand \doibase [0]{http://dx.doi.org/}%
\providecommand \selectlanguage [0]{\@gobble}%
\providecommand \bibinfo  [0]{\@secondoftwo}%
\providecommand \bibfield  [0]{\@secondoftwo}%
\providecommand \translation [1]{[#1]}%
\providecommand \BibitemOpen [0]{}%
\providecommand \bibitemStop [0]{}%
\providecommand \bibitemNoStop [0]{.\EOS\space}%
\providecommand \EOS [0]{\spacefactor3000\relax}%
\providecommand \BibitemShut  [1]{\csname bibitem#1\endcsname}%
\let\auto@bib@innerbib\@empty
\bibitem [{\citenamefont {Huang}(1987)}]{Huang}%
  \BibitemOpen
  \bibfield  {author} {\bibinfo {author} {\bibfnamefont {K.}~\bibnamefont
  {Huang}},\ }\href {https://books.google.fr/books?id=M8PvAAAAMAAJ} {\emph
  {\bibinfo {title} {Statistical mechanics}}}\ (\bibinfo  {publisher} {Wiley},\
  \bibinfo {address} {New York, USA},\ \bibinfo {year} {1987})\BibitemShut
  {NoStop}%
\bibitem [{\citenamefont {Parisi}\ \emph {et~al.}(1987)\citenamefont {Parisi},
  \citenamefont {{M. M{\'e}zard}},\ and\ \citenamefont {Virasoro}}]{MPV}%
  \BibitemOpen
  \bibfield  {author} {\bibinfo {author} {\bibfnamefont {G.}~\bibnamefont
  {Parisi}}, \bibinfo {author} {\bibnamefont {{M. M{\'e}zard}}}, \ and\
  \bibinfo {author} {\bibfnamefont {M.}~\bibnamefont {Virasoro}},\ }\href@noop
  {} {\emph {\bibinfo {title} {Spin Glass Theory and Beyond}}}\ (\bibinfo
  {publisher} {World Scientific},\ \bibinfo {address} {Singapore},\ \bibinfo
  {year} {1987})\BibitemShut {NoStop}%
\bibitem [{\citenamefont {M{\'e}zard}\ and\ \citenamefont
  {Parisi}(2001)}]{MP1}%
  \BibitemOpen
  \bibfield  {author} {\bibinfo {author} {\bibfnamefont {M.}~\bibnamefont
  {M{\'e}zard}}\ and\ \bibinfo {author} {\bibfnamefont {G.}~\bibnamefont
  {Parisi}},\ }\href@noop {} {\bibfield  {journal} {\bibinfo  {journal} {Eur.
  Phys. J. B}\ }\textbf {\bibinfo {volume} {20}},\ \bibinfo {pages} {217}
  (\bibinfo {year} {2001})}\BibitemShut {NoStop}%
\bibitem [{\citenamefont {M{\'e}zard}\ and\ \citenamefont
  {Parisi}(2003)}]{MP2}%
  \BibitemOpen
  \bibfield  {author} {\bibinfo {author} {\bibfnamefont {M.}~\bibnamefont
  {M{\'e}zard}}\ and\ \bibinfo {author} {\bibfnamefont {G.}~\bibnamefont
  {Parisi}},\ }\href@noop {} {\bibfield  {journal} {\bibinfo  {journal} {J.
  Stat. Phys.}\ }\textbf {\bibinfo {volume} {111}},\ \bibinfo {pages} {1}
  (\bibinfo {year} {2003})}\BibitemShut {NoStop}%
\bibitem [{\citenamefont {M\'ezard}\ and\ \citenamefont
  {Montanari}(2009)}]{MezMonbook}%
  \BibitemOpen
  \bibfield  {author} {\bibinfo {author} {\bibfnamefont {M.}~\bibnamefont
  {M\'ezard}}\ and\ \bibinfo {author} {\bibfnamefont {A.}~\bibnamefont
  {Montanari}},\ }\href@noop {} {\emph {\bibinfo {title} {Information, Physics,
  and Computation}}}\ (\bibinfo  {publisher} {Oxford University Press, Inc.},\
  \bibinfo {address} {New York, NY, USA},\ \bibinfo {year} {2009})\BibitemShut
  {NoStop}%
\bibitem [{\citenamefont {Yedidia}\ \emph {et~al.}(2005)\citenamefont
  {Yedidia}, \citenamefont {Freeman},\ and\ \citenamefont {Weiss}}]{yedidia}%
  \BibitemOpen
  \bibfield  {author} {\bibinfo {author} {\bibfnamefont {J.}~\bibnamefont
  {Yedidia}}, \bibinfo {author} {\bibfnamefont {W.~T.}\ \bibnamefont
  {Freeman}}, \ and\ \bibinfo {author} {\bibfnamefont {Y.}~\bibnamefont
  {Weiss}},\ }\href@noop {} {\bibfield  {journal} {\bibinfo  {journal} {IEEE T.
  Inform. Theory}\ }\textbf {\bibinfo {volume} {51}},\ \bibinfo {pages} {2282}
  (\bibinfo {year} {2005})}\BibitemShut {NoStop}%
\bibitem [{\citenamefont {Rizzo}\ \emph {et~al.}(2010)\citenamefont {Rizzo},
  \citenamefont {Lage-Castellanos}, \citenamefont {Mulet},\ and\ \citenamefont
  {Ricci-Tersenghi}}]{tommaso_CVM}%
  \BibitemOpen
  \bibfield  {author} {\bibinfo {author} {\bibfnamefont {T.}~\bibnamefont
  {Rizzo}}, \bibinfo {author} {\bibfnamefont {A.}~\bibnamefont
  {Lage-Castellanos}}, \bibinfo {author} {\bibfnamefont {R.}~\bibnamefont
  {Mulet}}, \ and\ \bibinfo {author} {\bibfnamefont {F.}~\bibnamefont
  {Ricci-Tersenghi}},\ }\href@noop {} {\bibfield  {journal} {\bibinfo
  {journal} {J. Stat. Phys.}\ }\textbf {\bibinfo {volume} {139}},\ \bibinfo
  {pages} {375} (\bibinfo {year} {2010})}\BibitemShut {NoStop}%
\bibitem [{\citenamefont {Lage-Castellanos}\ \emph {et~al.}(2013)\citenamefont
  {Lage-Castellanos}, \citenamefont {Mulet}, \citenamefont {Ricci-Tersenghi},\
  and\ \citenamefont {Rizzo}}]{average_mulet}%
  \BibitemOpen
  \bibfield  {author} {\bibinfo {author} {\bibfnamefont {A.}~\bibnamefont
  {Lage-Castellanos}}, \bibinfo {author} {\bibfnamefont {R.}~\bibnamefont
  {Mulet}}, \bibinfo {author} {\bibfnamefont {F.}~\bibnamefont
  {Ricci-Tersenghi}}, \ and\ \bibinfo {author} {\bibfnamefont {T.}~\bibnamefont
  {Rizzo}},\ }\href@noop {} {\bibfield  {journal} {\bibinfo  {journal} {J.
  Phys. A: Math. Theor.}\ }\textbf {\bibinfo {volume} {46}},\ \bibinfo {pages}
  {135001} (\bibinfo {year} {2013})}\BibitemShut {NoStop}%
\bibitem [{\citenamefont {Xiao}\ and\ \citenamefont {Zhou}(2011)}]{zhou_11}%
  \BibitemOpen
  \bibfield  {author} {\bibinfo {author} {\bibfnamefont {J.-Q.}\ \bibnamefont
  {Xiao}}\ and\ \bibinfo {author} {\bibfnamefont {H.}~\bibnamefont {Zhou}},\
  }\href {http://stacks.iop.org/1751-8121/44/i=42/a=425001} {\bibfield
  {journal} {\bibinfo  {journal} {J. Phys. A: Math. Theor.}\ }\textbf {\bibinfo
  {volume} {44}},\ \bibinfo {pages} {425001} (\bibinfo {year}
  {2011})}\BibitemShut {NoStop}%
\bibitem [{\citenamefont {Zhou}\ and\ \citenamefont {Wang}(2012)}]{zhou_12}%
  \BibitemOpen
  \bibfield  {author} {\bibinfo {author} {\bibfnamefont {H.}~\bibnamefont
  {Zhou}}\ and\ \bibinfo {author} {\bibfnamefont {C.}~\bibnamefont {Wang}},\
  }\href@noop {} {\bibfield  {journal} {\bibinfo  {journal} {J. Stat. Phys.}\
  }\textbf {\bibinfo {volume} {148}},\ \bibinfo {pages} {513} (\bibinfo {year}
  {2012})}\BibitemShut {NoStop}%
\bibitem [{\citenamefont {Dom\'inguez}\ \emph {et~al.}(2015)\citenamefont
  {Dom\'inguez}, \citenamefont {Lage-Castellanos},\ and\ \citenamefont
  {Mulet}}]{eduardoRFIM}%
  \BibitemOpen
  \bibfield  {author} {\bibinfo {author} {\bibfnamefont {E.}~\bibnamefont
  {Dom\'inguez}}, \bibinfo {author} {\bibfnamefont {A.}~\bibnamefont
  {Lage-Castellanos}}, \ and\ \bibinfo {author} {\bibfnamefont
  {R.}~\bibnamefont {Mulet}},\ }\href@noop {} {\bibfield  {journal} {\bibinfo
  {journal} {J. Stat. Mech.: Theor. Exp.}\ }\textbf {\bibinfo {volume} {2015}}
  (\bibinfo {year} {2015})}\BibitemShut {NoStop}%
\bibitem [{\citenamefont {Tanaka}\ and\ \citenamefont
  {Morita}(1995{\natexlab{a}})}]{tanakaCVM1995}%
  \BibitemOpen
  \bibfield  {author} {\bibinfo {author} {\bibfnamefont {K.}~\bibnamefont
  {Tanaka}}\ and\ \bibinfo {author} {\bibfnamefont {T.}~\bibnamefont
  {Morita}},\ }\href {\doibase DOI: 10.1016/0375-9601(95)00387-I} {\bibfield
  {journal} {\bibinfo  {journal} {Phys. Lett. A}\ }\textbf {\bibinfo {volume}
  {203}},\ \bibinfo {pages} {122 } (\bibinfo {year}
  {1995}{\natexlab{a}})}\BibitemShut {NoStop}%
\bibitem [{\citenamefont {Tanaka}\ \emph {et~al.}(2003)\citenamefont {Tanaka},
  \citenamefont {Inoue},\ and\ \citenamefont {Titterington}}]{tanakaCVM}%
  \BibitemOpen
  \bibfield  {author} {\bibinfo {author} {\bibfnamefont {K.}~\bibnamefont
  {Tanaka}}, \bibinfo {author} {\bibfnamefont {J.}~\bibnamefont {Inoue}}, \
  and\ \bibinfo {author} {\bibfnamefont {D.~M.}\ \bibnamefont {Titterington}},\
  }in\ \href@noop {} {\emph {\bibinfo {booktitle} {XIII —Proceedings of the
  2003 IEEE Signal Processing Society Workshop (17-19 September, 2003}}}\
  (\bibinfo  {publisher} {IEEE Computer Society Press},\ \bibinfo {year}
  {2003})\ pp.\ \bibinfo {pages} {329--338}\BibitemShut {NoStop}%
\bibitem [{\citenamefont {Yasuda}\ \emph {et~al.}(2015)\citenamefont {Yasuda},
  \citenamefont {Kataoka},\ and\ \citenamefont {Tanaka}}]{tanaka_replica_cvm}%
  \BibitemOpen
  \bibfield  {author} {\bibinfo {author} {\bibfnamefont {M.}~\bibnamefont
  {Yasuda}}, \bibinfo {author} {\bibfnamefont {S.}~\bibnamefont {Kataoka}}, \
  and\ \bibinfo {author} {\bibfnamefont {K.}~\bibnamefont {Tanaka}},\
  }\href@noop {} {\bibfield  {journal} {\bibinfo  {journal} {CoRR}\ }\textbf
  {\bibinfo {volume} {abs/1503.04585}} (\bibinfo {year} {2015})}\BibitemShut
  {NoStop}%
\bibitem [{\citenamefont {Gouillart}\ \emph {et~al.}(2013)\citenamefont
  {Gouillart}, \citenamefont {Krzakala}, \citenamefont {M\'ezard},\ and\
  \citenamefont {Zdeborov\'a}}]{gouillart_discrete_tom}%
  \BibitemOpen
  \bibfield  {author} {\bibinfo {author} {\bibfnamefont {E.}~\bibnamefont
  {Gouillart}}, \bibinfo {author} {\bibfnamefont {F.}~\bibnamefont {Krzakala}},
  \bibinfo {author} {\bibfnamefont {M.}~\bibnamefont {M\'ezard}}, \ and\
  \bibinfo {author} {\bibfnamefont {L.}~\bibnamefont {Zdeborov\'a}},\ }\href
  {http://stacks.iop.org/0266-5611/29/i=3/a=035003} {\bibfield  {journal}
  {\bibinfo  {journal} {Inverse Problems}\ }\textbf {\bibinfo {volume} {29}},\
  \bibinfo {pages} {035003} (\bibinfo {year} {2013})}\BibitemShut {NoStop}%
\bibitem [{\citenamefont {Tanaka}\ and\ \citenamefont
  {Morita}(1995{\natexlab{b}})}]{tanaka_morita_CVM_image95}%
  \BibitemOpen
  \bibfield  {author} {\bibinfo {author} {\bibfnamefont {K.}~\bibnamefont
  {Tanaka}}\ and\ \bibinfo {author} {\bibfnamefont {T.}~\bibnamefont
  {Morita}},\ }\href@noop {} {\bibfield  {journal} {\bibinfo  {journal}
  {Physics Letters A}\ }\textbf {\bibinfo {volume} {203}},\ \bibinfo {pages}
  {122} (\bibinfo {year} {1995}{\natexlab{b}})}\BibitemShut {NoStop}%
\bibitem [{\citenamefont {Sibel}\ and\ \citenamefont {Reynal}(2012)}]{sibel12}%
  \BibitemOpen
  \bibfield  {author} {\bibinfo {author} {\bibfnamefont {J.-C.}\ \bibnamefont
  {Sibel}}\ and\ \bibinfo {author} {\bibfnamefont {S.}~\bibnamefont {Reynal}},\
  }\href@noop {} {\bibfield  {journal} {\bibinfo  {journal} {International
  Conference on Control, Automation and Information Sciences}\ ,\ \bibinfo
  {pages} {4}} (\bibinfo {year} {2012})}\BibitemShut {NoStop}%
\bibitem [{\citenamefont {Pakzad}\ and\ \citenamefont
  {Anantharam}(2005)}]{Pakzad2005Kikuchi}%
  \BibitemOpen
  \bibfield  {author} {\bibinfo {author} {\bibfnamefont {P.}~\bibnamefont
  {Pakzad}}\ and\ \bibinfo {author} {\bibfnamefont {V.}~\bibnamefont
  {Anantharam}},\ }\href@noop {} {\bibfield  {journal} {\bibinfo  {journal}
  {Neural Computation}\ }\textbf {\bibinfo {volume} {17}},\ \bibinfo {pages}
  {1836} (\bibinfo {year} {2005})}\BibitemShut {NoStop}%
\bibitem [{\citenamefont {Wang}\ and\ \citenamefont
  {Zhou}(2013)}]{zhou_redundant}%
  \BibitemOpen
  \bibfield  {author} {\bibinfo {author} {\bibfnamefont {C.}~\bibnamefont
  {Wang}}\ and\ \bibinfo {author} {\bibfnamefont {H.-J.}\ \bibnamefont
  {Zhou}},\ }\href {http://stacks.iop.org/1742-6596/473/i=1/a=012004}
  {\bibfield  {journal} {\bibinfo  {journal} {Journal of Physics: Conference
  Series}\ }\textbf {\bibinfo {volume} {473}},\ \bibinfo {pages} {012004}
  (\bibinfo {year} {2013})}\BibitemShut {NoStop}%
\bibitem [{\citenamefont {Dom\'inguez}\ \emph {et~al.}(2011)\citenamefont
  {Dom\'inguez}, \citenamefont {Lage-Castellanos}, \citenamefont {Mulet},
  \citenamefont {Ricci-Tersenghi},\ and\ \citenamefont {Rizzo}}]{GBPGF}%
  \BibitemOpen
  \bibfield  {author} {\bibinfo {author} {\bibfnamefont {E.}~\bibnamefont
  {Dom\'inguez}}, \bibinfo {author} {\bibfnamefont {A.}~\bibnamefont
  {Lage-Castellanos}}, \bibinfo {author} {\bibfnamefont {R.}~\bibnamefont
  {Mulet}}, \bibinfo {author} {\bibfnamefont {F.}~\bibnamefont
  {Ricci-Tersenghi}}, \ and\ \bibinfo {author} {\bibfnamefont {T.}~\bibnamefont
  {Rizzo}},\ }\href {http://stacks.iop.org/1742-5468/2011/i=12/a=P12007}
  {\bibfield  {journal} {\bibinfo  {journal} {J. Stat. Mech.: Theor. Exp.}\
  }\textbf {\bibinfo {volume} {2011}},\ \bibinfo {pages} {P12007} (\bibinfo
  {year} {2011})}\BibitemShut {NoStop}%
\bibitem [{\citenamefont {Furtlehner}\ and\ \citenamefont
  {Decelle}(2016)}]{Aurelien16}%
  \BibitemOpen
  \bibfield  {author} {\bibinfo {author} {\bibfnamefont {C.}~\bibnamefont
  {Furtlehner}}\ and\ \bibinfo {author} {\bibfnamefont {A.}~\bibnamefont
  {Decelle}},\ }\href {\doibase 10.1007/s10955-016-1566-0} {\bibfield
  {journal} {\bibinfo  {journal} {Journal of Statistical Physics}\ }\textbf
  {\bibinfo {volume} {164}},\ \bibinfo {pages} {531} (\bibinfo {year}
  {2016})}\BibitemShut {NoStop}%
\bibitem [{\citenamefont {Montanari}\ and\ \citenamefont
  {Rizzo}(2005)}]{loop_corrected_tom}%
  \BibitemOpen
  \bibfield  {author} {\bibinfo {author} {\bibfnamefont {A.}~\bibnamefont
  {Montanari}}\ and\ \bibinfo {author} {\bibfnamefont {T.}~\bibnamefont
  {Rizzo}},\ }\href {http://stacks.iop.org/1742-5468/2005/i=10/a=P10011}
  {\bibfield  {journal} {\bibinfo  {journal} {J. Stat. Mech.: Theor. Exp.}\
  }\textbf {\bibinfo {volume} {2005}},\ \bibinfo {pages} {P10011} (\bibinfo
  {year} {2005})}\BibitemShut {NoStop}%
\bibitem [{\citenamefont {Yedidia}(2000)}]{idiosync_yedidia}%
  \BibitemOpen
  \bibfield  {author} {\bibinfo {author} {\bibfnamefont {J.~S.}\ \bibnamefont
  {Yedidia}},\ }in\ \href@noop {} {\emph {\bibinfo {booktitle} {Advanced mean
  field methods: Theory and practice}}},\ \bibinfo {editor} {edited by\
  \bibinfo {editor} {\bibfnamefont {D.}~\bibnamefont {Saad}}\ and\ \bibinfo
  {editor} {\bibfnamefont {M.}~\bibnamefont {Opper}}}\ (\bibinfo  {publisher}
  {The MIT Press},\ \bibinfo {year} {2000})\ Chap.~\bibinfo {chapter} {3},
  p.~\bibinfo {pages} {21}\BibitemShut {NoStop}%
\bibitem [{\citenamefont {Kikuchi}(1951)}]{kikuchi}%
  \BibitemOpen
  \bibfield  {author} {\bibinfo {author} {\bibfnamefont {R.}~\bibnamefont
  {Kikuchi}},\ }\href@noop {} {\bibfield  {journal} {\bibinfo  {journal} {Phys.
  Rev.}\ }\textbf {\bibinfo {volume} {81}},\ \bibinfo {pages} {988} (\bibinfo
  {year} {1951})}\BibitemShut {NoStop}%
\bibitem [{\citenamefont {Onsager}(1944)}]{onsager}%
  \BibitemOpen
  \bibfield  {author} {\bibinfo {author} {\bibfnamefont {L.}~\bibnamefont
  {Onsager}},\ }\href@noop {} {\bibfield  {journal} {\bibinfo  {journal} {Phys.
  Rev.}\ }\textbf {\bibinfo {volume} {368}},\ \bibinfo {pages} {117} (\bibinfo
  {year} {1944})}\BibitemShut {NoStop}%
\bibitem [{\citenamefont {Lage-Castellanos}\ \emph {et~al.}(2011)\citenamefont
  {Lage-Castellanos}, \citenamefont {Mulet}, \citenamefont {Ricci-Tersenghi},\
  and\ \citenamefont {Rizzo}}]{dual}%
  \BibitemOpen
  \bibfield  {author} {\bibinfo {author} {\bibfnamefont {A.}~\bibnamefont
  {Lage-Castellanos}}, \bibinfo {author} {\bibfnamefont {R.}~\bibnamefont
  {Mulet}}, \bibinfo {author} {\bibfnamefont {F.}~\bibnamefont
  {Ricci-Tersenghi}}, \ and\ \bibinfo {author} {\bibfnamefont {T.}~\bibnamefont
  {Rizzo}},\ }\href {\doibase 10.1103/PhysRevE.84.046706} {\bibfield  {journal}
  {\bibinfo  {journal} {Phys. Rev. E}\ }\textbf {\bibinfo {volume} {84}},\
  \bibinfo {pages} {046706} (\bibinfo {year} {2011})}\BibitemShut {NoStop}%
\bibitem [{\citenamefont {Barber}\ \emph {et~al.}(1985)\citenamefont {Barber},
  \citenamefont {Pearson}, \citenamefont {Toussaint},\ and\ \citenamefont
  {Richardson}}]{barber_ising_3d}%
  \BibitemOpen
  \bibfield  {author} {\bibinfo {author} {\bibfnamefont {M.~N.}\ \bibnamefont
  {Barber}}, \bibinfo {author} {\bibfnamefont {R.~B.}\ \bibnamefont {Pearson}},
  \bibinfo {author} {\bibfnamefont {D.}~\bibnamefont {Toussaint}}, \ and\
  \bibinfo {author} {\bibfnamefont {J.~L.}\ \bibnamefont {Richardson}},\
  }\href@noop {} {\bibfield  {journal} {\bibinfo  {journal} {Phys. Rev. B}\
  }\textbf {\bibinfo {volume} {32}},\ \bibinfo {pages} {1720} (\bibinfo {year}
  {1985})}\BibitemShut {NoStop}%
\bibitem [{\citenamefont {Gaunt}\ \emph {et~al.}(1979)\citenamefont {Gaunt},
  \citenamefont {Sykes},\ and\ \citenamefont
  {McKenzie}}]{gaunt1979susceptibility}%
  \BibitemOpen
  \bibfield  {author} {\bibinfo {author} {\bibfnamefont {D.}~\bibnamefont
  {Gaunt}}, \bibinfo {author} {\bibfnamefont {M.}~\bibnamefont {Sykes}}, \ and\
  \bibinfo {author} {\bibfnamefont {S.}~\bibnamefont {McKenzie}},\ }\href@noop
  {} {\bibfield  {journal} {\bibinfo  {journal} {Journal of Physics A:
  Mathematical and General}\ }\textbf {\bibinfo {volume} {12}},\ \bibinfo
  {pages} {871} (\bibinfo {year} {1979})}\BibitemShut {NoStop}%
\bibitem [{\citenamefont {Parisi}\ and\ \citenamefont
  {Ruiz-Lorenzo}(1996)}]{parisi_lorenzo_5d}%
  \BibitemOpen
  \bibfield  {author} {\bibinfo {author} {\bibfnamefont {G.}~\bibnamefont
  {Parisi}}\ and\ \bibinfo {author} {\bibfnamefont {J.~J.}\ \bibnamefont
  {Ruiz-Lorenzo}},\ }\href@noop {} {\bibfield  {journal} {\bibinfo  {journal}
  {Phys. Rev. B}\ }\textbf {\bibinfo {volume} {54}},\ \bibinfo {pages} {R3698}
  (\bibinfo {year} {1996})}\BibitemShut {NoStop}%
\end{thebibliography}%

\appendix

\section{Maximal GBP is always valid } \label{ap:mmmmvalid}

We will prove that independently of the regions chosen as maximal, the introduction of multiplicative messages from maximal 
regions to their children, as explained in the main text, generates a valid GBP.

Valid means that equations (\ref{eq:mc_r}) is satisfied.
First of all, let us prove that (\ref{eq:brgl}) satisfies equation (\ref{eq:mc_r}). Without loss of generality, let us focus on a 
given region $r_0\in R$ and one of its children regions $\alpha \subset r_0$. The message $m_{r_0\to\alpha}(\xg)$ appears in 
the beliefs equations of all regions $r$ such that $r_0\cap r = \alpha$, which defines:
\begin{eqnarray}
 \Bra \equiv  \partial \mra  &=& \{r' \in R|  r' \cap r = \alpha\} \nonumber
 \end{eqnarray} 
 An example of $\Bra$ are the regions (except $Q_1$) appearing in the diagrams of Fig. \ref{fig:cube3D_mQ1_s1}.
 
 The property we need to prove is (restating eq. (\ref{eq:mc_r}))
\begin{prop} \label{prop}
 \begin{equation}
  \forall_{r_0 \in R_0}\; \forall_{\substack{\alpha \in R\\ \alpha \subset r_0}} \:\: \sum_{r'\in \Bra } c_{r'} = 0 \label{eq:c_Bro} 
 \end{equation}
 \end{prop}
This property is similar to the one defining the counting numbers eq. (\ref{eq:c_r}), but not the same. We will show the 
validity of (\ref{eq:c_Bro})  from that of the eq. (\ref{eq:c_r}).  Let us start by re-stating eq. (\ref{eq:c_r}) as
\begin{equation}
\sum_{r\in \Ac_\alpha} c_r = 1. \label{eq:c_rap}
\end{equation}
where we have defined $\Ac_\alpha$ the extended set of ancestors of any region $\alpha$ to include $\alpha$ itself
\[   \Ac_\alpha \equiv A_\alpha \cup \{\alpha\}
\]
The set $\Bra \subset \Ac_\alpha$. Furthermore, we can write:
\[ \Ac_\alpha = \Bra + \Braco
\]
where the sets in the right hand side are disjoint, and  
\begin{defi}{$\Braco $}
\[ \Braco = \{ r\in R | r\cap r_0 > \alpha\} 
\]
is the set of all ancestors of $\alpha$ such that their intersection with $r_0$ is larger than $\alpha$. 
\end{defi}

The part $\Braco$ contains the ancestors of $\alpha$ in whose beliefs the message $\mra$ does not appear, and 
it will never be an empty set, since it includes at least $r_0$ and its ancestry. In figure \ref{fig:cube3D_mQ1_s1} $\Braco$
is exactly the absent part with respect to Fig. \ref{fig:cube3D}, including also $Q_1$.

From now on we will use relational operators $>,<,\geq,\leq$ freely, since the hierarchy stablished by the inclusion of sets in the 
regions defines a partially ordered set. In this sense, $\alpha < \beta$ means that $\alpha \subset \beta$ and $\alpha \neq \beta$.

Since the sets $\Bra, \Braco$ are disjoint, the sum (\ref{eq:c_rap}) can be split among them becoming
\begin{equation}
\sum_{r\in \Bra} c_r +  \sum_{r\in \Braco} c_r = 1. \label{eq:c_BB}
\end{equation}
If $\Braco = \Acrc$, then proposition \ref{prop} is proved, since the second sum will equal 1. Otherwise, the validity of the maximal 
messages (eq. (\ref{eq:c_Bro})) falls from proving that in the most general case
\begin{equation}
 \sum_{r\in \Braco} c_r = 1, \label{eq:c_Braco}
\end{equation}
as we will through the rest of the appendix.

Let us see some properties of the set $\Braco$. Form now on, all greek letters $\alpha, \beta, \gamma \ldots$ refer to regions in $\Braco$.
\begin{lem} \label{lem:poset}
$\Braco$ is a finite partially ordered set.
\end{lem}
\begin{proof}
Since $\Braco \subset R$, it is finite. Furthermore, the set of all regions $R$  itself is a partially ordered set, defined by the 
inclusion relation $r_1 < r_2 \Longleftrightarrow r_1\subset r_2 \wedge r_1 \neq r_2 \: \qedhere$. 
\end{proof}
From now on we will use the terminology of partially ordered sets. For instance, we will say that region $r_2$ covers $r_1$ if $r_2 > r_1$ 
and there is no $z$ such that $r_2 >z>r_1$.

\begin{lem} \label{lem:Bracoclosed}
$\Braco$ is closed under intersection with $r_0$.
\end{lem}
\begin{proof}
The set of regions generated by the cluster variational method is closed under intersections in general. Let $\gamma \in \Braco$,
then $\gamma > \alpha$ and $\gamma \cap r_0 = \eta >\alpha$. Then, trivially, $\eta \cap r_0 = \eta > \alpha$ which guarantees 
that $\eta \in \Braco \:\qedhere$. 
\end{proof}
Since $\eta < r_0$, the following corollaries follows.
\begin{cor}
All $\beta_i \in \Braco$ such that $\beta_i$ covers $\alpha$ inside $\Braco$, are subsets of $r_0$ ($\beta_i < r_0$).
\end{cor}
\begin{cor}
\[B = \{ \beta_i \in \Braco | \beta_i  \mbox{ covers } \alpha \mbox{ in } \Braco\} \] is the set of minimal elements in $\Braco$.
\end{cor}
\begin{lem} 
If $\gamma \in \Braco$  then $\Ac_{\gamma} \subset \Braco$. 
\end{lem}
\begin{proof}
\[ \forall_{\eta \in \Ac_\gamma} \eta \cap r_0 \geq \gamma >\alpha \Rightarrow \eta \in \Braco \: \qedhere\]
\end{proof}
As a consequence the entire $\Braco$ is generated by the ancestry of members of $B$, this is:
\begin{lem}\label{lem:Braco_in_A}
\[\Braco = \mathop{\bigcup}_{\beta_i\in B} \Ac_{\beta_i} \]
\end{lem}
We have written $\Braco$ in terms of the set of ancestors of some minimal elements $\beta_1,\beta_2\ldots$. We emphasize that 
all such ancestries share, at least, the common element $r_0$ and its ancestry. Therefore, they are not disjoint sets. 

In order to prove (\ref{eq:c_Braco}) we start from the fact that 
\[ \forall_{\beta\in R} \:\: \sum_{r\in \Ac_\beta} c_r =1 
\]
Furthermore, everytime that the intersection of ancestries is not empty, which is the case of all minimal elements $\beta_1,\beta_2\ldots$
since they all share $r_0$, there exists an element in $\gamma \in \Braco$, such that $\Ac_\gamma = \Ac_{\beta_i} \cap \Ac_{\beta_j}$. Let us formalize 
and generalize this idea.

\begin{figure*}[!htb]
 \includegraphics[width=\textwidth,height=4cm]{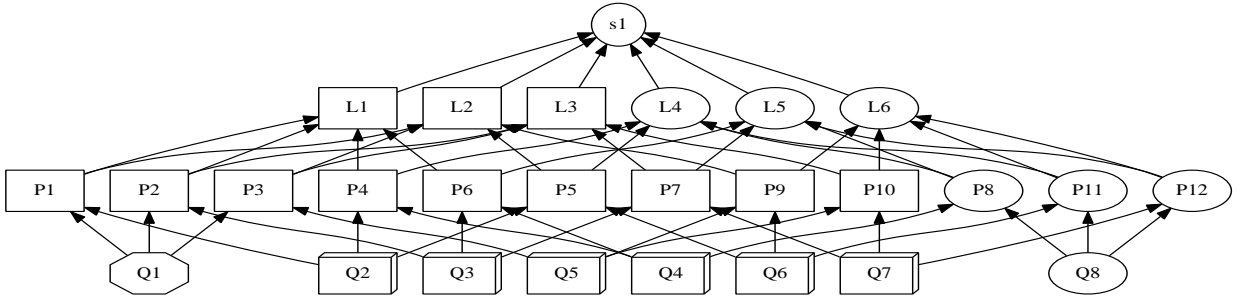}
\caption{Hasse Diagram of the regions in the 3D-cube approximation, for the 3D Ising model. Variables live in the 
nodes of a tridimensional cubic lattice. Maximal regions are the cubes (basic cell), and from them
all intersections generates new regions, as prescribed by cluster variational method. See also figure \ref{fig:cube3D}.\\
In this representation, all regions containing one central spin $s_1$ are depicted.
Central spin $s_1=\alpha$ is  surrounded by eight maximal cubic regions $Q_1,\ldots Q_8$. 
We represent the partially ordered set defined by the inclusion relations among the ancestors of $s_1$. Arrows point in the
Parent-to-Child direction.\\
Bottom figures are the eight cubes, next level are all the square plaquettes that intersect among these cubes, 
the third layer is made of the six links containing spin $s_1$. \\
We consider $Q_1 = r_0$ to be the region sending message to spin $s_1$. Shapes correspond to the position in the ancestry of $s_1$ with respect to 
the message $m_{Q_1\to s_1}$. All circular regions  represent elements of $\Bra$, and therefore their intersection with $Q_1$ is exactly $s_1$ 
(represented also in Fig. \ref{fig:cube3D_mQ1_s1}, except for $Q_1$). 
Angular regions are those in $\Braco$ (the absent part of Fig. \ref{fig:cube3D_mQ1_s1}, including $Q_1$). \label{fig:hasse}}
\end{figure*}

Let us use the definition of least upper bound $\phi = \lub(\gamma,\eta)$ as the smallest element in $\Braco$ that is both $\phi \geq \gamma$ and 
$\phi \geq \eta$. By minimum we mean that every other $z$ that shares both properties, happens to be $z>\eta$. In finite posets, the least 
upper bound might not exist, but if it does, it is uniquely defined.

We will show that the intersections of the ancestries of elements in $\Braco$, can be written 
themselves as the ancestry of another element in $\Braco$.

\begin{lem}[The intersection of ancestries]
Let ${\gamma_1,\gamma_2\in\Braco}$ with a non null intersection of their ancestors $ \Ac_{\gamma_1} \cap \Ac_{\gamma_2}  \neq\varnothing$ then, 
\begin{eqnarray*}
&& \exists  \eta = \lub(\gamma_1,\gamma_2) \in \Braco \\
\mbox{such that  } &&\: \Ac_\eta =\Ac_{\gamma_1} \cap \Ac_{\gamma_2} 
\end{eqnarray*}
\end{lem}
\begin{proof}
If the intersection of ancestries is not null, then it has at least one element, let us say $\theta$. Since $\gamma_1,\gamma_2 \leq \theta$, 
then all ancestors of $\theta$ are also comparable and above $\gamma_1$ and $\gamma_2$. In other words, 
$\theta \in \Ac_{\gamma_1} \cap \Ac_{\gamma_2}   \Rightarrow \Ac_\theta \subset  \Ac_{\gamma_1} \cap \Ac_{\gamma_2}$. Let us suppose 
that $ \Ac_{\gamma_1} \cap \Ac_{\gamma_2} $ is the union of more than one ancestries of many incomparable  $\theta_i$'s. This cannot 
be the case, since the intersection of any two $\theta_1$ and $\theta_2$ produces a lower $\theta$, that is again in 
$ \Ac_{\gamma_1} \cap \Ac_{\gamma_2} $ and whose ancestry includes that of both $\theta_1$ and $\theta_2$. Repeating this procedure, 
we end up with a unique value $\eta$, such that $ \Ac_\eta = \Ac_{\gamma_1} \cap \Ac_{\gamma_2}  \: \qedhere$. The fact that 
$\eta= \lub(\gamma_1,\gamma_2)$ is trivial.
\end{proof}

Since the set $\Braco$ is written as the union of ancestries of the minimal sets $\beta_1,\beta_2,\ldots$ (see lemma \ref{lem:Braco_in_A}), then from 
the previous lemma and the fact that all ancestries of the minimal elements share at least $r_0$ and its ancestry, it follows that
\begin{cor}
 Let $\Pi$ be any subset of the indices of the minimal sets $\beta_1,\beta_2,\ldots$ covering  $\alpha$ inside $\Braco$, then
 \[\exists_{r\in \Braco} \:\:\bigcap_{i\in\Pi} \Ac_{\beta_i} = \Ac_{r}\]
\end{cor}
 
Now, using the inclusion exclusion principle, we can write 
 \begin{eqnarray*}
 \sum_{r\in \Braco} c_r& =& \sum_{\beta\in B} \sum_{\gamma \in \Ac_\beta } c_\gamma - \\
 && - \sum_{\beta_1,\beta_2 \in B} \sum_{\gamma \in \Ac_{\beta_1}\cap \Ac_{\beta_2} } c_\gamma + \\
 && + \sum_{\beta_1,\beta_2, \beta_3 \in B} \sum_{\gamma \in \Ac_{\beta_1}\cap \Ac_{\beta_2}\cap \Ac_{\beta_3}} c_\gamma - \\
 && \ldots
\end{eqnarray*}
where every second sumatory is itself of the form
\[\sum_{r \in \Ac_\gamma } c_r = 1\]
by the previous corollary. So, if the set $\Braco$ is generated by the ancestries of $K$ minimal elements $\beta_i$, then
 \begin{eqnarray*}
 \sum_{r\in \Braco} c_r& =& \binom{K}{1}-\binom{K}{2} + \binom{K}{3} - \\
 && \ldots +(-1)^{K+1} \binom{K}{K}\\
 & =& 1+(1-1)^K = 1
\end{eqnarray*}
which concludes the proof of (\ref{eq:c_Braco}) and from it that of (\ref{eq:c_Bro}).$\blacksquare$
 
To help visualize a little bit the rather complicated algebra of sets, we depict in Fig. \ref{fig:hasse} the Hasse diagram and the sets $\Braco$ and $\Bra$ 
for the Cube-CVM approximation for the 3D square lattice spin model. In particular the case of the message going from the cubic region $Q_1$ to the 
central spin region $s_1$ is shown.

\section{Marginalization of beliefs}
\label{ap:beliefs_rotation}
Extremization of the variational free energy with respect to the message functions result in the following equation
\begin{equation}
 \sum_{r\in\partial m_{r_0\to\gamma} } c_r\sum_{\underline{x}_{r}\setminus \underline{x}_{\gamma}} b_r(\xr) = 0 
 \label{eq:c_r_b_r_eq_zero}
 \end{equation}
for each message $m_{r_0\to\gamma}$.

In this appendix we show that the only solution for such set of equations are beliefs satisfying
\begin{equation} \label{eq:consistency_app}
\displaystyle \forall_{p \in A_r}  \: b_r(\xr) = \sum_{\underline{x}_{p}\setminus \xr}  b_{p}. (\underline{x}_{p}) 
\end{equation}

We will prove this in an inductive manner, assuming that starting from the maximal regions onto a certain level, all regions marginalize
and from this assumption, we will show that the next level also correctly marginalize.

{ \bf Induction: base case }

Let $r_0$ be a region whose ancestry $\Arc$ consists only on maximal regions. 
Then the intersection of members of any two members of $\Arc$ cannot be smaller than $r_0$ since by definition $r\in \Arc \Rightarrow r_0\subset r$, 
but the intersection cannot be larger than $r_0$, since in such case $\gamma = r_1 \cap r_2 > r_0$ would be an ancestor of $r_0$ that is not 
a maximal region. So, any two elements of $\Arc$ intersect exactly at $r_0$. 

Given the definition of the set of regions in whose belief a given message is present
\[ \partial m_{p\to r} = \{r' | r' \cap p = r\}
\] 
the previous result implies that any message $m_{r_1\to r_0}(\xr_0)$ from a parent $r_1$ of $r_0$ appears in the belief 
of all the other parents except $r_1$ itself and including $r_0$. Mathematically:
\begin{eqnarray*} 
\forall_{r\in \Arc}  \partial m_{r\to r_0} &=& \{r_0\} \cup \Arc \setminus \{r\} \\ 
 &=&  \Acrc \setminus \{r\} \\ 
\end{eqnarray*}
Let there be $|\Arc| = K$ parents to $r_0$. We have $K$ equations of the type (\ref{eq:c_r_b_r_eq_zero})
\[
\displaystyle \forall_{i\in[1,2,\ldots,K]}  \sum_{r\in \Acrc\setminus \{r_i\} } c_r\sum_{\underline{x}_{r}\setminus \underline{x}_0} b_r(\xr) = 0 
 \]
 Except for $c_{r_0} =1-K$, all other counting numbers are $c_r=1$. Furthermore, we can add the missing sumand in each case, to obtain, for each $i\in[1,2,\ldots K]$
\[
 (1-K) b_0(\underline{x}_0) + \sum_{r\in \Arc}\sum_{\underline{x}_{r}\setminus \underline{x}_0} b_r(\xr) =  \sum_{\underline{x}_{i}\setminus \underline{x}_0} b_{r_i}(\underline{x}_i)
 \]
 Now, the left hand side is independent of $i$, and therefore all righthand sides have to be equal for every $i$. This proves that the parents are consistent
 among each other on their belief at region $r_0$. To show that they agree with $b_0(\underline{x}_0)$, we now use the fact that they agree to write
\[
 (1-K) b_0(\underline{x}_0) + K \sum_{\underline{x}_{r_i}\setminus \underline{x}_0} b_{r_i}(\underline{x}_i) =  \sum_{\underline{x}_{i}\setminus \underline{x}_0} b_{r_i}(\underline{x}_i)
 \]
 which concludes the induction base case with
\[
\displaystyle  \forall_{i\in[1,\ldots,K]} \:\: b_0(\underline{x}_0) =  \sum_{\underline{x}_{i}\setminus \underline{x}_0} b_{r_i}(\underline{x}_i).
 \]
 as desired.
 
{\bf Induction: inductive step}
 
Focus on a given region $r_0$, and consider its ancestry $\Arc$. In a partial order, the ancestry of a given element
is always generated by the union of ancestries of all elements covering it (see the previous appendix for the definition of the cover). Let there be $K$ such elements $r_i$ covering $r_0$, then
\[ \Arc = \bigcup_{1\leq i\leq K} A^o_{r_i}
\]
The induction step assumes that all ancestors $p\in A_{r_i}$ of $r_i$  are consistent with $r_i$, in the sense of (\ref{eq:consistency_app}), 
and will then prove that
$r_i$ has also to be consistent with $r_0$. Since consistency between any $p\in A_{r_i}$ and $r_i$ and between $r_i$ and $r_0$ is given, transitivity implies
consistency between  $r_0$ and $p$, and generalizing, with all the ancestry $\Arc$ of $r_0$, concluding the induction step.

\begin{figure*}[!htb]
  \includegraphics[width=0.65\textwidth]{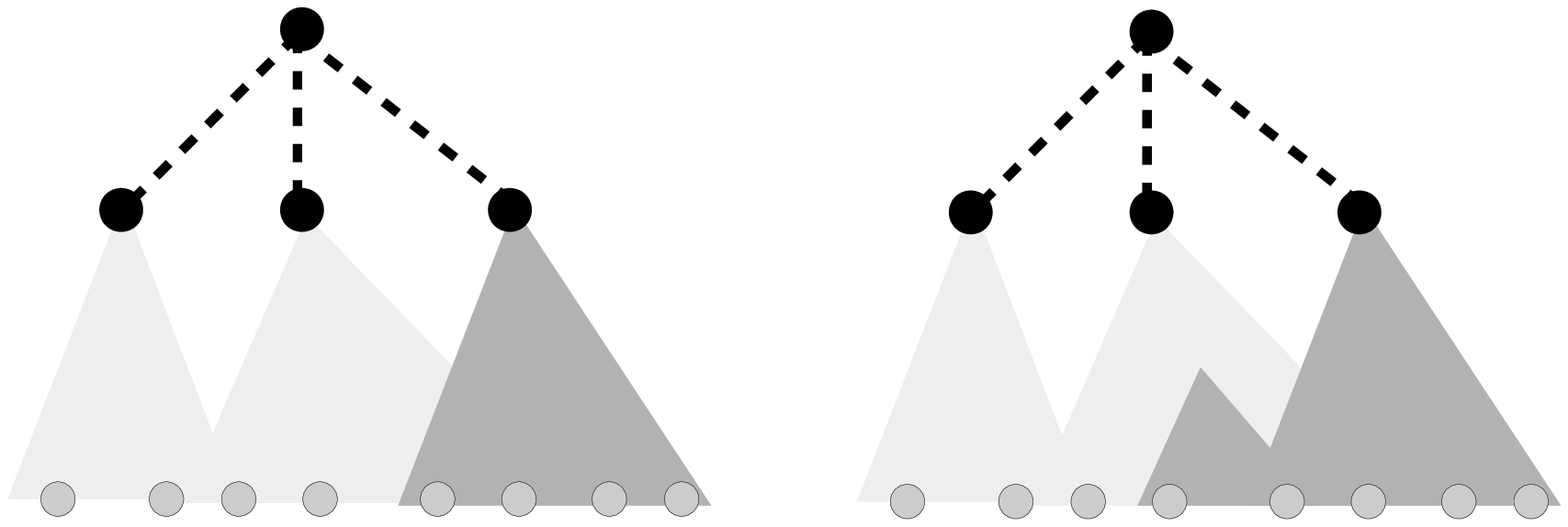}
\caption{
\label{fig:ancestries}}
\end{figure*}

The tricky part is to show that consistency of the cover elements $r_i$ with their ancestry $A_{r_i}$ implies consistency of $r_i$ with $r_0$. 
In order to do so let us start by the following 
\begin{lem}[Intersection of ancestries]
 If the set $\{r_1,\ldots ,r_K\}$ covers the element $r_0$, then for any two distincts $k_1$ and $k_2$ in $[1,\ldots,K]$,  
 if there exists $p\in A_{r_{k_1}}$ such that $\gamma = p\cap r_{k_2} > r_0$, then $\gamma = r_{k_2}$. 
\end{lem} 
In other words, any element $p\in \Arc$ such that the intersection $p\cap r_i$ with one of the covers $r_i$ is larger than $r_0$, is itself an ancerstor of that cover $p \in A^o_{r_i}$.
\begin{proof}
 It is enough to note that $\gamma = p\cap r_{k_2} > r_0$ is nessarily bounded $r_0< \gamma \leq r_{k_2}$, but since $r_{k_2}$ is a cover of $r_0$, no such intermediate element can exists, and therefore the only accepted situation is $\gamma = r_{k_2}$. But this also implies  that $r_{k_2} = p\cap r_{k_2}$
 which means that $p$ is ancestor of $k_2$.
\end{proof}

\begin{cor}
 The set $\partial m_{r_i\to r_0} = \{p\in \Acrc | p\cap r_i = r_0\}$ is given by
 \[\partial m_{r_i\to r_0}  = \Acrc \setminus A^o_{r_i}\]
\end{cor}
Graphically this means that the only possible situation for the sets $\partial m_{r_i\to r_0}$ is the one in the left panel of figure \ref{fig:ancestries} and the situation in the right is forbiden.

Now, for every cover region $r_i, i\in[1,\dots,K]$ we will have an equation (\ref{eq:c_r_b_r_eq_zero}), that using the previous corollary can be written as
\begin{equation}
 \sum_{r\in \Acrc} c_r\sum_{\underline{x}_{r} \setminus \underline{x}_0} b_r(\xr) =  \sum_{r\in A^o_{r_i}} c_r \sum_{\underline{x}_{r}\setminus \underline{x}_0} b_{r}(\underline{x}_r) \label{eq:sum_cr_consistency}
 \end{equation}
 
Now it is quite similar to the base case of the induction. The left hand side does not depends on $i\in[1,\ldots,K]$, while the right hand does. Therefore any two $i_1,i_2 \in [1,\ldots,K]$ will be consistent
\[\sum_{r\in A^o_{r_{i_1}}} c_r \sum_{\underline{x}_{r}\setminus \underline{x}_0} b_{r}(\underline{x}_r) =\sum_{r\in A^o_{r_{i_2}}} c_r \sum_{\underline{x}_{r}\setminus \underline{x}_0} b_{r}(\underline{x}_r)\]
Furthermore, using that $c_{r_i} = 1- \sum_{r\in A_{r_i}}c_r$ and the consistency of $r_i$ with its ancestry, the previous equality can be transformed in
\[\sum_{\underline{x}_{r_{i_1}}\setminus \underline{x}_0} b_{r_{i_1}}(\underline{x}_{r_{i_1}}) =\sum_{\underline{x}_{r_{i_2}}\setminus \underline{x}_0} b_{r_{i_2}}(\underline{x}_{r_{i_2}}) 
\]
Proceeding in a similar fashion as done in the base case for the induction, we can use the consistency between all different $r_i$ back in equation (\ref{eq:sum_cr_consistency}) to show that they also have to agree with $r_0$
\[b_{r_0}(\underline{x}_{r_0}) =\sum_{\underline{x}_{r_{i}}\setminus \underline{x}_0} b_{r_{i}}(\underline{x}_{r_{i}}) 
\]
concluding the inductive step. $\blacksquare$.

\section{Moment matching fields are gauge free}
\label{ap:mmmmgauge}
Let us prove theorems \ref{theo:mmmmconsistent} and \ref{theo:mmmmgauge}. They both say that using Maximal Messages with Moment Matching fields guarantees
the consistency of beliefs, and is gauge free. We have already shown in the previous appendix \ref{ap:mmmmvalid} that extremization of the CVM free energy
with respect to message functions, ensures consistency of the beliefs. However, when using Moment Matching fields, we reduce the degrees of freedom of the
message functions, and therefore it is not clear that consistency still holds. 

The reduction on the active fields was done seeking a gauge free parametrization of the variational free energy. 
So we would also like to prove that the solution to the extremization problem gives a unique solution to the fields defining the
messages. In other words, that when doing moment matching, we have removed all fields except those necessary to guarantee the consistency between the beliefs.

The proof of the consistency of the beliefs (theorem \ref{theo:mmmmconsistent}) will be carried in the following way:
\begin{enumerate}
 \item Show that extremization with respect to moment matching fields on a set of variables $\prod_{i\in q} s_i$ ensure consistency of the 
 corresponding moments $\langle \prod_{i\in q}s_i\rangle$ among those beliefs containing that group of variables.
 \item Show that all moments are fixed by some field.
 \item Conclude by saying that if all moments are equal, distributions have to be consistent.
\end{enumerate}

The proof of the gauge free character (theorem  \ref{theo:mmmmgauge}) will be carried out simply by showing that, after
removing of the undesired fields, there are as many variables (fields) as equations to be solved.

\subsection{Moment consistency (theorem \ref{theo:mmmmconsistent})}

The prescription given by the moment matching definition \ref{def:moment_matching} says that message $m_{p\to r}(\sr) $ 
counts with a field $U_{q}$ forcing the correlation among variables in $q\subseteq r$ as  
\[m_{p\to r}(\sr)  = e^{\cdots+U^{p\to r}_q \prod_{i\in q} s_i + \cdots}
\] 
if and only if $r$ is the smallest region containing the subset $q$. Therefore, for any given set $q=\{s_{q_1},\ldots,s_{q_k}\}$ there is a single region
$r_0$ such that the fields $U^{p\to r_0}_q$ appear in the message from its ancestors $p\in A_{r_0}$. The situation is depicted in figure \ref{fig:mmmm_diagram}.

\begin{figure*}[!htb]
  \includegraphics[width=0.65\textwidth]{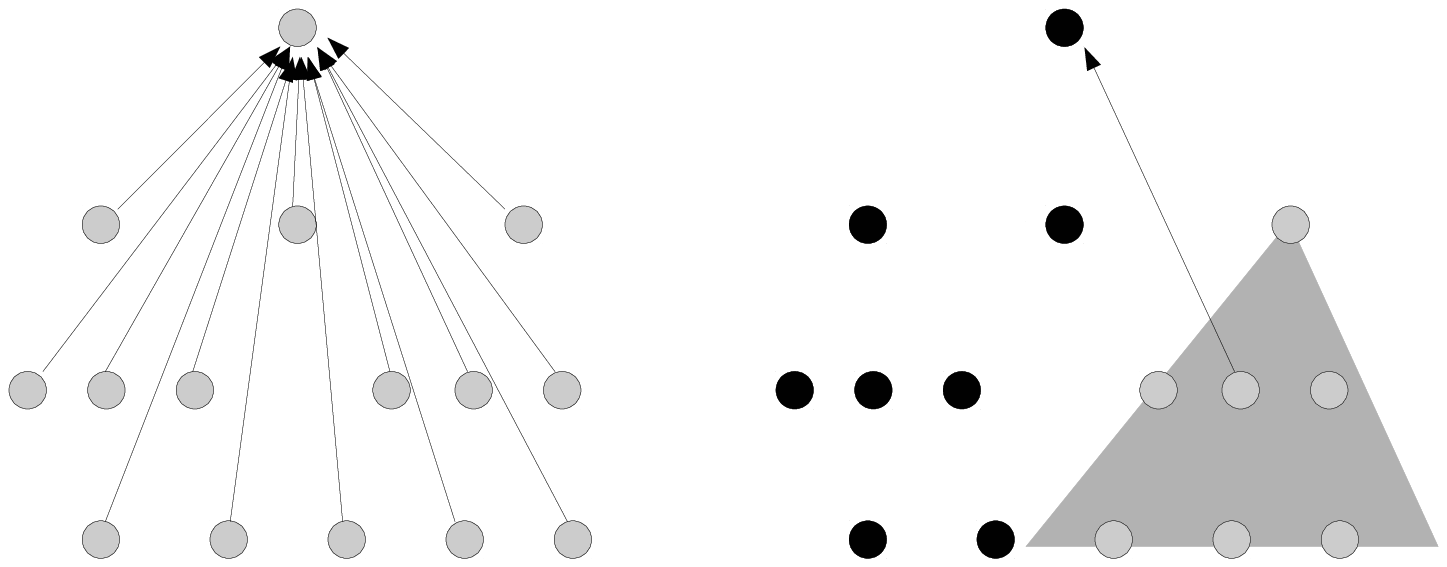}
\caption{
\label{fig:mmmm_diagram}}
\end{figure*}

When extremizing with respect to the fields $U$ that parametrize the messages we get equations similar to (\ref{eq:c_r_b_r_eq_zero}), but
instead of the belief functions, we get the moments corresponding to the field
\begin{eqnarray}
\forall_{p\in \Acrc} \sum_{r\in\partial U^{p\to r}_q } c_r \: \xi_{q,r}\: &=& 0 \label{eq:sum_cr_corr}\\
 \mbox{with}\phantom{xx}  \xi_{q,r} &=& \left\langle \prod_{i\in q}x_i \right\rangle_{b_r} \nonumber
 \end{eqnarray}
 and the expected value is taken with respect of the belief $b_r(\xr)$ (local distribution) at region $r$ and 
 $\partial U^{p\to r}_q = \partial m_{p\to r}$ is the set of regions in whose beliefs the message $m_{p\to r}$ participates, and therefore 
so does $U^{p\to r}_q$. 

We will have as many equations (\ref{eq:sum_cr_corr}) as ancestors $r_0$ have $K=|A_{r_0}|$, corresponding to derivation with respect to each 
field (each arrow in figure \ref{fig:mmmm_diagram}). We will further assume that no counting number is zero, and therefore the set of linear 
equations (\ref{eq:sum_cr_corr}) relates all moments $\xi_{q,r}$ which are $K+1$, including the one obtained at $r_0$ itself.

Since we have proven in Appendix \ref{ap:mmmmvalid} that 
\begin{eqnarray}
 \sum_{r\in\partial U^{p\to r}_q } c_r &=& 0 
\end{eqnarray}
a particular solution of the system of equations will always be 
\[ \forall_{r\in \Arc} \:\: \xi_{q,r} = \xi_{q,r_0}
\]
which is part of what we are trying to prove. We can also write equation (\ref{eq:sum_cr_corr}) as
\[\forall_{p\in \Arc} \sum_{r\in\partial U^{p\to r}_q \setminus r_0 } c_r \: \xi_{q,r}\: = - c_{r_0} \xi_{q,r_0}
\]
where the right hand side is the same irrespective of who is $p$. Considering only the correlations involved in the left hand side,
this system of equation will have only one solution (at fixed $\xi_{q,r_0}$) if the matrix $G_{K\times K}$ made of elements
\[ g_{r,p} =\left\{ \begin{array}{cc}
   c_{r} & \mbox{ if }r\in\partial U^{p\to r_0}_q \\
   0 & \mbox{ otherwise }
  \end{array}\right.
\]
has non zero determinant.

In order not to make the paper far too long, we ask the reader to prove in each cluster variational method implemented that the 
set of regions used fulfill this property. Yet, we warn that this property wont be fulfilled any time that some of the following conditions
is present:
\begin{itemize}
 \item There are zero counting numbers, since a column of the matrix will be full of zeros.
 \item A given message does not appear in any belief equation, since a row will be full of zeros.
 \item Two or more rows of the matrix are equal, causing a zero determinant.
\end{itemize}
Without a proof (that seems rather complicated to us), we give the hint that these seems to be the only situations possible, after many random playing 
with arbitrary clusters approximations. It is not obvious why two or more lines could not be linearly combined into another line, to cause a zero 
determinant, but some properties of the counting numbers seems to forbid this.

So, we have that under the condition of non-singular matrix $G_{K\times K}$ the set of equations involving the correlation of a 
given set of spins $q$, force all such correlations to be equal. Now since every set of spins contained in two or more regions is contained in their
intersection (which also has to be a region by CVM prescription), then all two regions agree on the moments of every common subset of variables.

We finish the proof by noting that if all regions agree on all correlations of the intersecting variables, they have to be consistent in the sense that the 
marginal probabilities over these variables should agree.

\subsection{Gauge free (theorem \ref{theo:mmmmgauge})}

We note from the previous proof that we have $K$ fields $U^{p\to r_0}_q$ if there are $K$ ancestor of region $r_0$. Since the consistency of the 
corresponding correlations $\xi_{q,p}$ is fixed by $K$ equations, we have as many parameters as equations to be satisfied. Furthermore, the consistency
among the local distributions $b_r(\xr)$ that contain a given set of variables $q$, can not be forced with less than $K$ equalities, since equalities are 
transitive and therefore all we need is to conect the set of $K+1$ regions containing $q$ in a graph with minimal number of edges (each edge meaning an 
equality) among moments. Among $K+1$ nodes, the single component graph with minimal edges is the tree, which happens to have $K$ edges. So, as we said,
the $K$ fields $U^{p\to r_0}_q$ are exactly the minimun required amount to enforce all local distribution to agree on the respective moment  $\xi_{q,p}$.

This used not to be the case in Parent-to-Child CVM, for instance, as seen in  \cite{GBPGF}, where the consistency of a single spin magnetization
that belonged to four links and four plaquettes in the square plaquette Ising model, appear after the derivation with respect to 
twelve parameters (field) instead of eight. Therefore there are 12 equations (and 12 parameters) to asign the equality among 8 local distributions,
forcing 4 of the equations (and parameters) to be redundant.

So, in this appendix we have proven that maximal messages and moment matching fields generate a set of equations with the following properties:
\begin{enumerate}
 \item every correlation among a set of variables that belongs to two or more regions is present in some equation;
 \item there are as many free parameters as relations required to guarantee consistency;
 \item consistent correlations is one solution of the system.
\end{enumerate}

 

\end{document}